\documentclass{article}
\pdfoutput=1

\usepackage[preprint, nonatbib]{neurips_2021}

\usepackage[utf8]{inputenc} 
\usepackage[T1]{fontenc}    
\usepackage{hyperref}       
\usepackage{url}            
\usepackage{booktabs}       
\usepackage{amsfonts}       
\usepackage{nicefrac}       
\usepackage{microtype}      
\usepackage{xcolor}         
\usepackage{kotex}
\usepackage{amsmath}
\usepackage{amsthm}
\usepackage{graphicx}
\usepackage{multirow}
\usepackage{tikz}
\usepackage{pgfplots}
\usepackage{longtable}
\usepackage{siunitx}
\usepackage{subfig}
\usepackage{verbatim}
\usepackage{amssymb}
\usepackage{wrapfig}
\usepackage{enumitem}
\usepackage[normalem]{ulem}
\usepackage{hyperref}
\useunder{\uline}{\ul}{}

\newtheorem{definition}{Definition}
\newtheorem{theorem}{Theorem}
\newtheorem{lemma}{Lemma}
\newcommand{\bb}[1]{\mathbf{#1}}
\definecolor{darkblue}{rgb}{0,0.0,0.7}
\definecolor{Gray}{gray}{0.9}

\newcommand{\md}[1]{{\color{black} #1}}
\newcommand{\me}[1]{{\color{black} #1}}
\newcommand{\mf}[1]{{\color{black} #1}}
\newcommand{\mg}[1]{{\color{black} #1}}

\DeclareMathOperator{\sgn}{sgn}
\DeclareMathOperator{\ReLU}{ReLU}

\title{Precise Approximation of Convolutional Neural Networks for Homomorphically Encrypted Data}

\author{
  Junghyun Lee\thanks{Equal contribution.}$^{\:\:1}$, Eunsang Lee${}^{*1}$, Joon-Woo Lee\thanks{Corresponding authors.} ${}^{\:1}$,\\
  \textbf{Yongjune Kim}$^{\dagger 2}$ \textbf{, Young-Sik Kim}$^3$ \textbf{, Jong-Seon No}${}^1$
  \\
  \\
  $^1$Department of Electrical and Computer Engineering, INMC, Seoul National University \\
  $^2$Department of Information and Communication Engineering, DGIST
  \\
  $^3$Department of Information and Communication Engineering, 
  Chosun University\\
  \texttt{ljhfree530@snu.ac.kr, shaeunsang@snu.ac.kr, joonwoo42@snu.ac.kr,} \\
  \texttt{yjk@dgist.ac.kr,
  iamyskim@chosun.ac.kr,
  jsno@snu.ac.kr} \\
}

\begin{document}

\maketitle
\begin{abstract}
\md{Homomorphic encryption is one of the representative solutions to privacy-preserving machine learning (PPML) classification \mf{enabling} the server to classify private data of clients while guaranteeing privacy. This work focuses on PPML using word-wise fully homomorphic encryption (FHE)}. In order to implement deep learning on \md{word-wise homomorphic encryption (HE)}, the ReLU and max-pooling functions should be approximated by some polynomials for homomorphic operations. \mf{Most of the previous studies focus on HE-friendly networks, where the ReLU and max-pooling functions are approximated using low-degree polynomials.} However, for the classification of the CIFAR-10 dataset, using a low-degree polynomial requires designing a new deep learning model and training. \mf{In addition, this approximation by low-degree polynomials cannot support deeper neural networks due to large approximation errors.} Thus, we propose a precise polynomial approximation technique for the ReLU and max-pooling functions. \mf{Precise approximation using a single polynomial requires an exponentially high-degree polynomial, which results in a significant number of non-scalar multiplications}. Thus, we propose a method to approximate the ReLU and max-pooling functions accurately using a \emph{composition of minimax approximate polynomials of small degrees}. If we replace the ReLU and max-pooling functions with the proposed approximate polynomials, \md{well-studied} deep learning models such as ResNet and VGGNet can still be used without further modification for PPML on FHE. Even pre-trained parameters can be used without retraining. \mg{We} approximate the ReLU and max-pooling functions in the ResNet-152 using the composition of minimax approximate polynomials of degrees 15, 27, and 29. \mg{Then,} we succeed in classifying the plaintext ImageNet dataset with 77.52\% accuracy, which is very close to the original model accuracy of 78.31\%.

\end{abstract}

\section{Introduction}

\md{One of the most promising areas in privacy-preserving machine learning (PPML) is deep learning using homomorphic encryption (HE).} HE is a cryptographic system that allows algebraic operations such as addition and multiplication on encrypted data. Using HE, the server can perform deep learning on the encrypted data from the client while fundamentally preventing the privacy leakage of private data. \mg{Recently,} deep learning using HE has been studied \mg{extensively} \cite{gilad2016cryptonets, gazelle, sealion, ngraph, FasterCryptoNets, cheetah, HCNN_GPU}. These studies can be classified into deep learning studies using only HE and those using both HE and multi-party computation (MPC) techniques together. Our work focuses on \mg{the former}, which has advantages of low communication cost and \mg{no risk of model exposure} compared to hybrid methods using both HE and MPC techniques. In particular, we focus on deep learning using the word-wise fully homomorphic encryption (FHE), which allows the addition and multiplication of an encrypted array over the set of complex numbers $\mathbb{C}$ or $\mathbb{Z}_p$ for an integer $p>2$ without restriction on the number of operations.


In word-wise FHE, the ReLU and max-pooling functions cannot be performed on the encrypted data because they are non-arithmetic operations. \mg{Thus, after the approximate polynomials replace the ReLU and max-pooling functions, deep learning on the encrypted data can be performed. Most of the previous studies of deep learning on HE focus on HE-friendly networks, where the ReLU and max-pooling functions are approximated using low-degree polynomials.} These studies on HE-friendly networks \me{have} some limitations.
\mg{For instance, designing a new deep learning model and training are required in \cite{FasterCryptoNets,HCNN_GPU,codespy}, which classified the CIFAR-10 \cite{krizhevsky2009learning}.} However, in many real-world applications, training data or training frameworks are not accessible, and thus, the previously studied low-degree approximate polynomials cannot be used in these applications. Furthermore, HE-friendly networks cannot support very deep neural networks due to large approximation errors. \mg{On the other hand, the computational time and accuracy of bootstrapping in word-wise FHE have recently been improved significantly \cite{GPUCKKS,lee2020optimal, variance, scale_invariant}, allowing the evaluation of higher-degree polynomials for many times. Thus, it is essential to study precise approximate polynomials.}

\subsection{Our contributions}
\begin{figure}
    \centering
    \includegraphics[width=\linewidth]{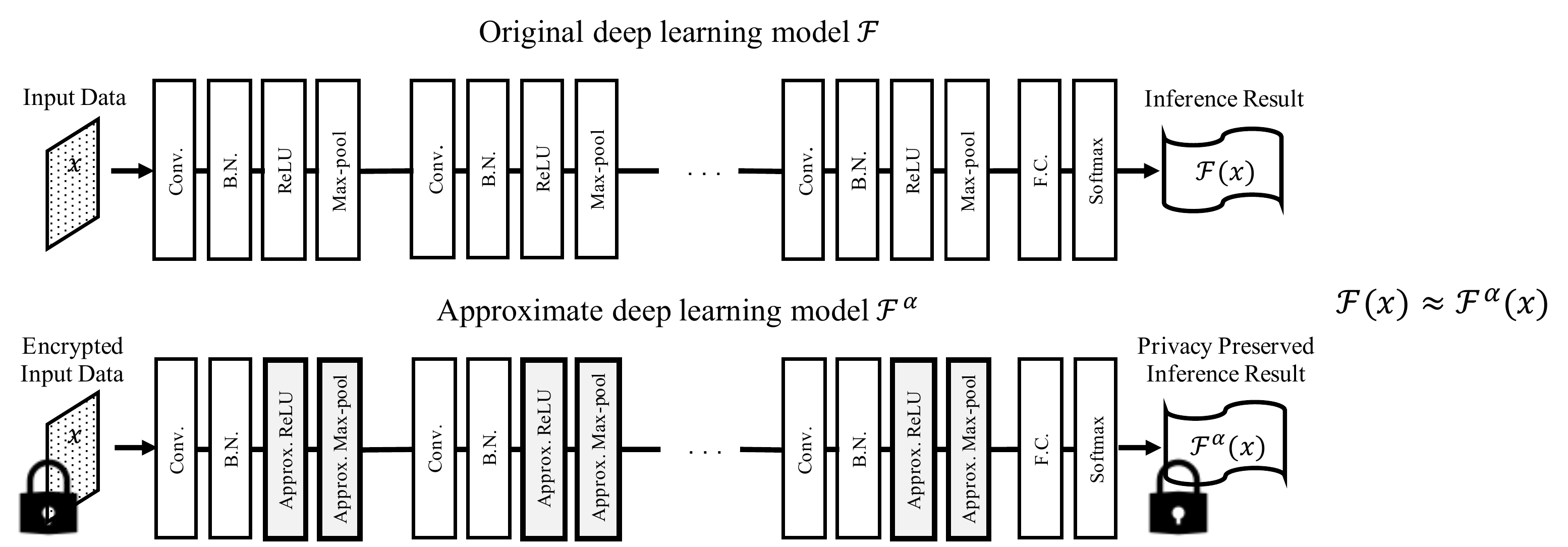}
    \caption{Comparison between the original deep learning model $\mathcal{F}$ and the proposed approximate deep learning model $\mathcal{F}^\alpha$ with the precision parameter $\alpha$. The proposed model is obtained by replacing the ReLU and max-pooling functions in the original models with the proposed approximate polynomials.}
    \vspace{-0.5cm}
    \label{fig:my_label}
\end{figure}

\md{Precise approximation requires a polynomial of degree at least about $2^{1.0013\alpha-2.8483}$ for the precision \me{parameter} $\alpha$, which results in $2^{\Theta(\alpha)}$ number of non-scalar multiplications and several numerical issues (see Section 3.3).}
Thus, we propose a \md{precise polynomial approximation technique for the ReLU and max-pooling functions that uses} \emph{a composition of minimax approximate polynomials of small degrees}. Hence, we can replace the ReLU and max-pooling functions with the proposed approximate polynomials in well-studied deep learning models such as ResNet \cite{he2016deep}, VGGNet \cite{simonyan2014very}, and GoogLeNet \cite{szegedy2015going}, and the pre-trained parameters \mf{can be} used without retraining. 
Figure \ref{fig:my_label} shows the comparison between the original deep learning model $\mathcal{F}$ and the proposed approximate deep learning model $\mathcal{F}^\alpha$ for the precision parameter $\alpha$, where $\mathcal{F}$ is the function that outputs the inference result for an input $\bb{x}$.

We prove that the two inference results, $\mathcal{F}^\alpha (\bb{x})$ and $\mathcal{F}(\bb{x})$, satisfy $||\mathcal{F}^\alpha (\bb{x}) - \mathcal{F}(\bb{x})||_\infty \leq C 2^{-\alpha}$ for a constant $C$, \mg{which can be determined independently of the precision parameter $\alpha$} (see Theorem 4). Theoretically, if we increase the precision parameter $\alpha$ enough, the two inference results \mf{would be identical}. The numerical evaluation also supports that the two inference results are very close for an appropriate $\alpha$.

We classify the CIFAR-10 without retraining using well-studied deep learning models \mf{such as ResNet and VGGNet}, and pre-trained parameters.
\mf{We can also further reduce the required precision parameter $\alpha$ by allowing} \me{retraining for only a few epochs.}
For classifying the ImageNet \cite{russakovsky2015imagenet} \me{in} ResNet-152, we approximate the ReLU and max-pooling functions using the composition of minimax approximate polynomials of degrees 15, 27, and 29. Then, \md{for the first time for PPML using word-wise HE}, we successfully classify the ImageNet with 77.52\% accuracy, which is very close to the original model accuracy of 78.31\%. The proposed approximate polynomials are generally applicable to various deep learning models such as ResNet, VGGNet, and GoogLeNet for PPML on FHE. 



\subsection{Related works}
Deep learning on the homomorphically encrypted data, i.e., PPML using HE, can be classified into two categories: One \me{uses HE only}, and the other uses HE and MPC techniques together. Gazelle \cite{gazelle}, Cheetah \cite{cheetah}, and nGraph-HE \cite{ngraph} are representative studies that perform deep learning using HE and MPC techniques together. They used garbled circuits to perform activation functions through communication with clients and servers. However, the use of MPC techniques can lead to high communication costs or \md{disclose information on the configuration of the deep learning models.}

\md{PPML using only HE can be classified into PPML using bit-wise HE and PPML using word-wise HE. In \cite{she}, \me{PPML using the fast fully homomorphic encryption over the torus (TFHE) \cite{TFHE}, which is the most promising bit-wise FHE, was} studied. This \mf{PPML} can evaluate the ReLU and max-pooling functions exactly on the encrypted data. \mf{However, leveled HE was used instead of FHE because of the large computational time of FHE operations, which limits performing a very deep neural network on HE.} Our work focuses on deep learning using word-wise FHE \me{such as the Brakerski/Fan-Vercauteren \cite{BFV} and Cheon-Kim-Kim-Song (CKKS) scheme \cite{CKKS}}. CryptoNets \cite{gilad2016cryptonets}, Faster CryptoNets \cite{FasterCryptoNets}, SEALion \cite{sealion}, and studies in \cite{HCNN_GPU, codespy} also performed deep learning using word-wise HE. By designing a new deep learning model and training, the authors in \cite{FasterCryptoNets}, \cite{HCNN_GPU}, and \cite{codespy} classified the CIFAR-10 with an accuracy of about $75\%$, $77\%$, and $91.5\%$, respectively. CryptoNets and SEALion replaced the ReLU function with square function $x^2$, and Faster CryptoNets replaced the ReLU function with quantized minimax approximate polynomial $2^{-3}x^2 + 2^{-1}x + 2^{-2}$.}  

\section{Preliminaries}
\subsection{Minimax composite polynomial}
Unfortunately, FHE schemes do not support non-arithmetic operations such as sign function, where $\sgn(x) = x/|x|$ for $x \neq 0$, and $0$, otherwise. Hence, instead of the \md{actual} sign function, a polynomial \mg{approximating} the sign function should be evaluated on the encrypted data.  Recently, authors in \cite{leeminimax} proposed a composite polynomial named \emph{minimax composite polynomial} that optimally approximates the sign function with respect to the depth consumption and the number of non-scalar multiplications. The definition of \emph{minimax composite polynomial} is given as follows:

\begin{definition}
Let $D$ be $[-b,-a] \cup [a,b]$. A composite polynomial $p_k \circ p_{k-1} \circ \cdots \circ p_1$ is called a minimax composite polynomial on $D$ for $\{d_i\}_{1 \leq i \leq k}$ if the followings are satisfied:  
\begin{itemize}
\item $p_1$ is the minimax approximate polynomial of degree at most $d_1$ on $D$ for $\mathrm{sgn}(x)$.
\item For $2 \leq i \leq k$, $p_i$ is the minimax approximate polynomial of degree at most $d_i$ on $p_{i-1} \circ p_{i-2} \circ \cdots \circ p_1 (D)$ for $\mathrm{sgn}(x)$.
\end{itemize}
\end{definition}

\mg{For $\epsilon, 0<\epsilon<1$ and $\beta>0$, a polynomial $p(x)$ that approximates $\sgn(x)$ is called $(\beta,\epsilon)$-close if it satisfies $|p(x) - \text{sgn}(x)| \leq 2^{-\beta}$ for $x \in [-1,-\epsilon]\cup[\epsilon,1]$ \cite{compg}. }Then, the optimal $(\beta,\epsilon)$-close composite polynomial $p(x)$ is the minimax composite polynomial on $[-1,-\epsilon]\cup[\epsilon,1]$ for $\{d_i\}_{1 \leq i \leq k}$, where $\{d_i\}_{1 \leq i \leq k}$ is obtained from Algorithm 5 for inputs $\beta$ and $\epsilon$ in \cite{leeminimax}.

\subsection{Homomorphic max function}\label{sec:homomorphic_max_function}

To perform the max function on the encrypted data, called the homomorphic max function, a polynomial that approximates the max function should be evaluated on the encrypted data instead of the max function. The approximate polynomial of the max function, $m(a,b)$, should satisfy the following \mf{inequality} for the precision parameter $\alpha>0$:
\begin{equation}\label{max_error_condition}
|m(a,b) - \max(a,b)| \leq 2^{-\alpha} ~~~ \mbox{for }a,b \in [0,1].
\end{equation}

In \cite{compg}, considering $\max(a,b) = \frac{(a+b)+(a-b) \sgn(a-b)}{2}$, they obtained a polynomial $p(x)$ that approximates $\sgn(x)$ and then used it to determine a polynomial $m(a,b)$ that approximates $\max(a,b)$, that is, 
\begin{equation}\label{approx_maxsgn}
m(a,b) = \frac{(a+b)+(a-b) p(a-b)}{2}.
\end{equation}

\begin{wraptable}{R}{6cm}
\vspace{-0cm}
\centering
\caption{The optimized max function factor, degrees of the component polynomials, and depth consumption for evaluating $m_\alpha(a,b)$ according to the precision parameter $\alpha$ \cite{leeminimax}. }
\label{maxfactor}
\footnotesize{
\begin{tabular}{c|ccc}
\toprule
$\alpha$ & $\zeta_\alpha$ & degrees & depth \\ \hline
4        & 5                           & \{5\}                                                              & 4     \\
5        & 5                           & \{13\}                                                             & 5     \\
6        & 10                          & \{3,7\}                                                            & 6     \\
7        & 11                          & \{7,7\}                                                            & 7     \\
8        & 12                          & \{7,15\}                                                           & 8     \\
9        & 13                          & \{15,15\}                                                          & 9     \\
10       & 13                          & \{7,7,13\}                                                         & 11    \\
11       & 15                          & \{7,7,27\}                                                         & 12    \\
12       & 15                          & \{7,15,27\}                                                        & 13    \\
13       & 16                          & \{15,15,27\}                                                       & 14    \\
14       & 17                          & \{15,27,29\}                                                       & 15   \\
\bottomrule
\end{tabular}
}
\vspace{-0.8cm}
\end{wraptable}

In \cite{leeminimax}, \mg{the authors proposed a minimax composite polynomial $p(x)$ approximating the sign function. Then, they proposed a polynomial} $m(a,b)$ approximating the max function using $p(x)$ and equation (\ref{approx_maxsgn}).
Specifically, for a precision parameter $\alpha$ and \mg{the} optimized max function factor $\zeta_\alpha >0$, an $(\alpha-1,\zeta_\alpha \cdot 2^{-\alpha})$-close minimax composite polynomial $p_\alpha(x) = p_{\alpha,k} \circ \cdots \circ p_{\alpha,1}(x)$ \mf{was} obtained by using the optimal degrees obtained from Algorithm 5 in \cite{leeminimax}, \mg{where the depth consumption is minimized.} Then, it was shown that $m_\alpha(a,b) = \frac{(a+b)+(a-b) p_\alpha(a-b)}{2}$ from $p_\alpha(x)$ satisfies the error condition in equation (\ref{max_error_condition}). Table \ref{maxfactor} lists the values of $\zeta_\alpha$, corresponding depth consumption for evaluating $m_\alpha(a,b)$, and degrees of the component polynomials $p_{\alpha,1}, \cdots, p_{\alpha,k}$ according to $\alpha$. \mf{This approximate polynomial $m_{\alpha}(a,b)$ has the best performance up to now.}

\section{Precise polynomial approximation of the ReLU and max-pooling functions}
It is possible to approximate the ReLU and max-pooling functions by Taylor polynomials or minimax approximate polynomials \cite{Remez}. \mg{However, these approximate polynomials with minimal approximation errors require many non-scalar multiplications.} For this reason, we approximate the ReLU and max-pooling functions using $p_\alpha(x)$ and $m_\alpha(a,b)$, \mf{which are defined in Section \ref{sec:homomorphic_max_function}}. \mg{Although we only deal with the ReLU function among many activation functions, the proposed approximation method can also be applied to other activation functions such as the leaky ReLU function.}

\subsection{Precise approximate polynomial of the ReLU function}
In this subsection, we propose a polynomial that precisely approximates the ReLU function, a key component of deep neural networks. For a given precision parameter $\alpha$, it is required that the approximate polynomial of the ReLU function, $r(x)$, satisfies the following error condition:
\begin{equation}\label{ReLU_error_condition}
|r(x) - \ReLU(x)| \leq 2^{-\alpha} ~~~ \mbox{for }x \in [-1,1].
\end{equation}

\begin{figure}[b]%
\centering
\begin{minipage}[b]{.2\textwidth}%
\subfloat[]{\includegraphics[width=\linewidth]{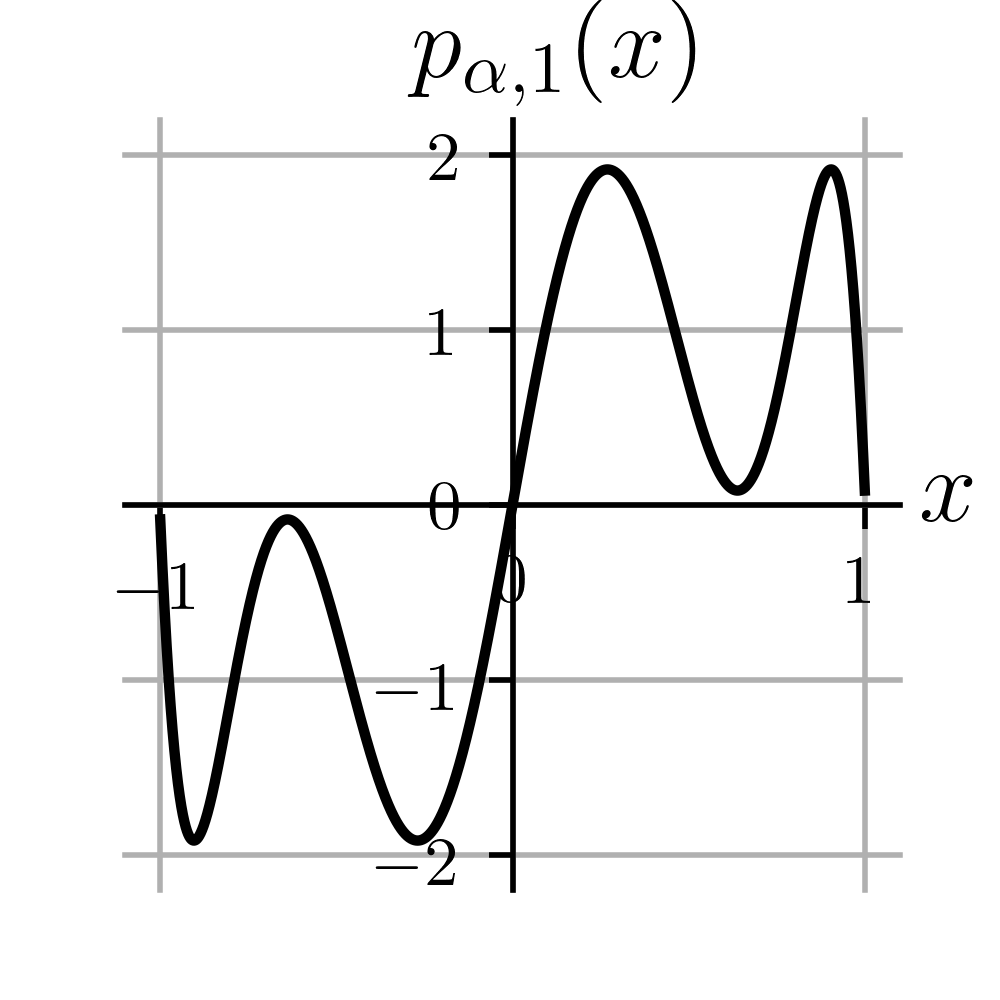}}
\end{minipage}%
\begin{minipage}[b]{.2\textwidth}%
\subfloat[]{\includegraphics[width=\linewidth]{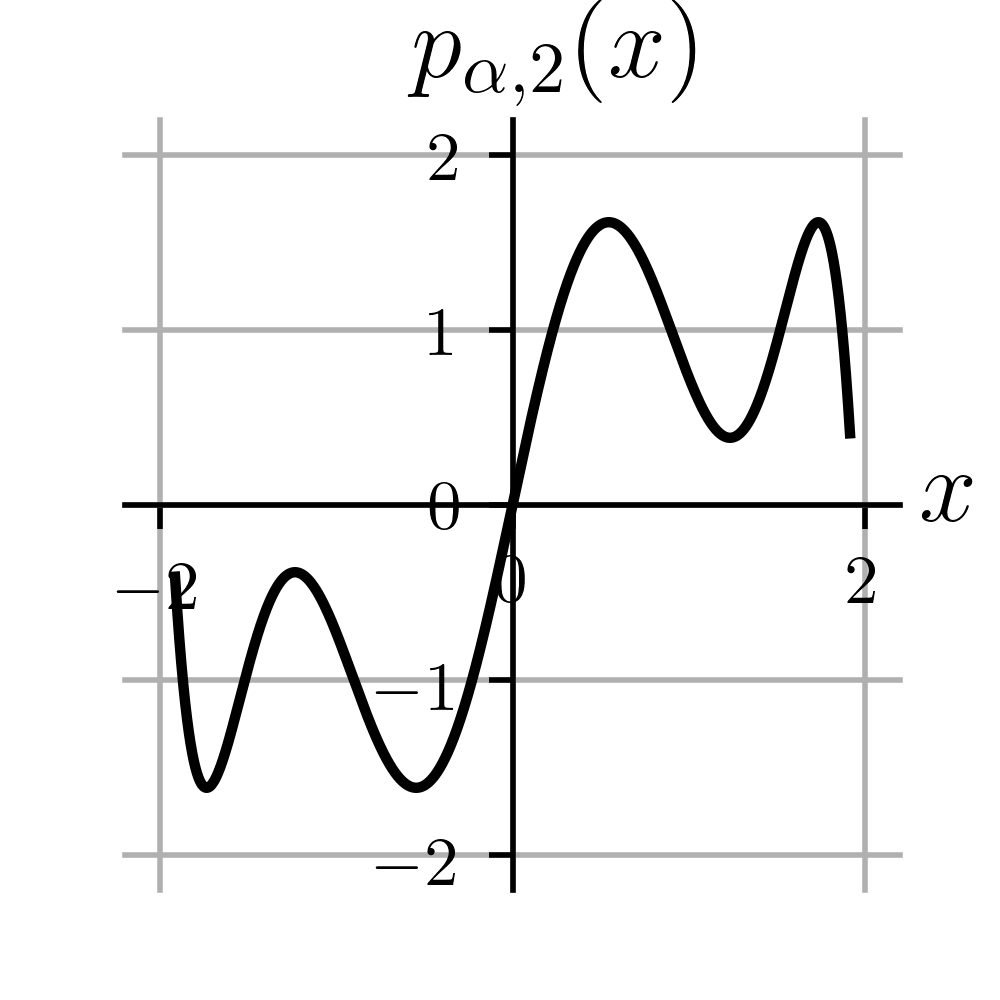}}
\end{minipage}%
\begin{minipage}[b]{.2\textwidth}%
\subfloat[]{\includegraphics[width=\linewidth]{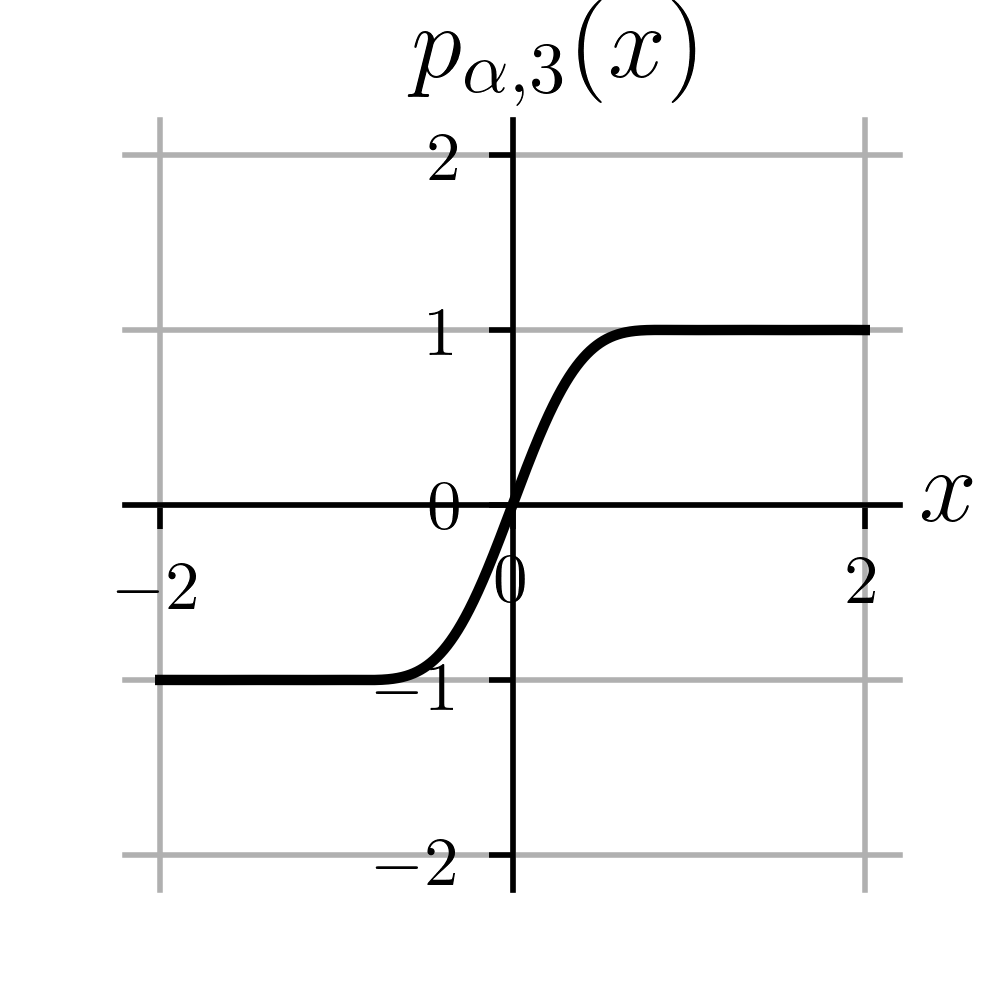}}
\end{minipage}%
\begin{minipage}[b]{.2\textwidth}%
\subfloat[]{\includegraphics[width=\linewidth]{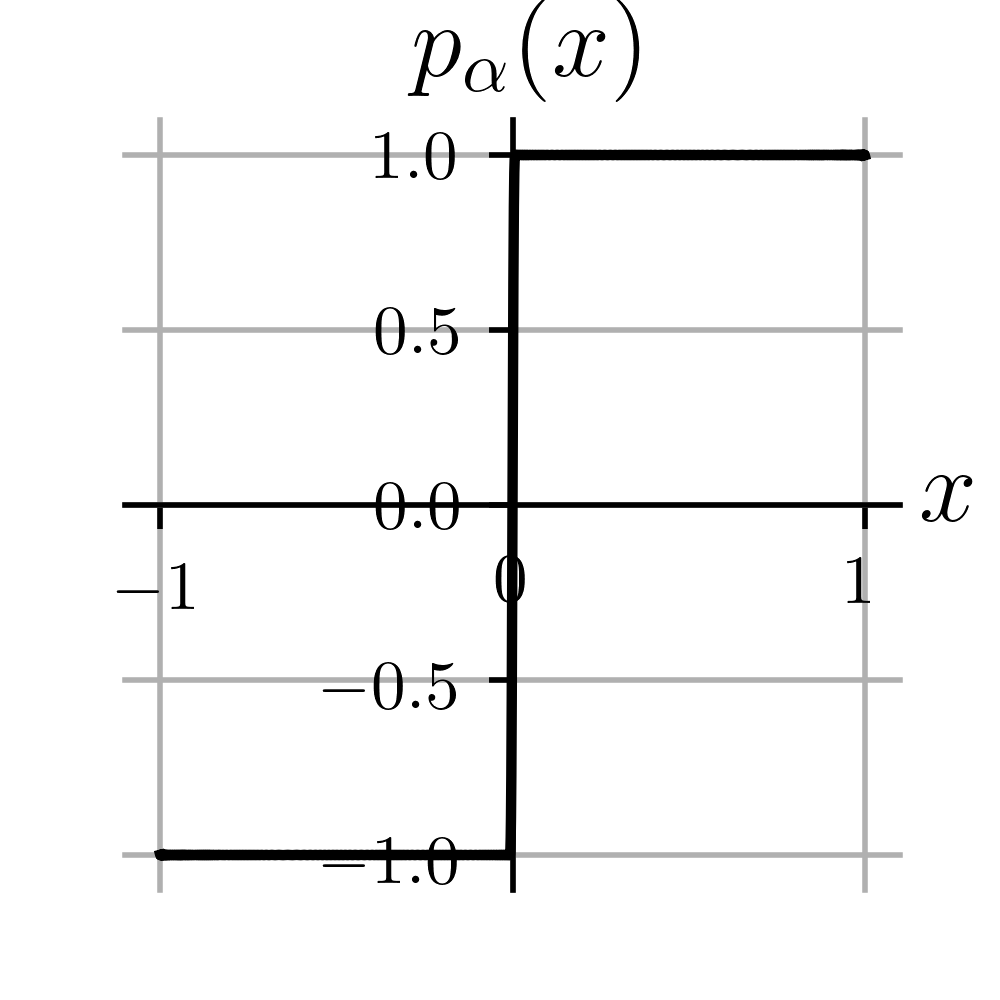}}
\end{minipage}%
\begin{minipage}[b]{.2\textwidth}%
\subfloat[]{\includegraphics[width=\linewidth]{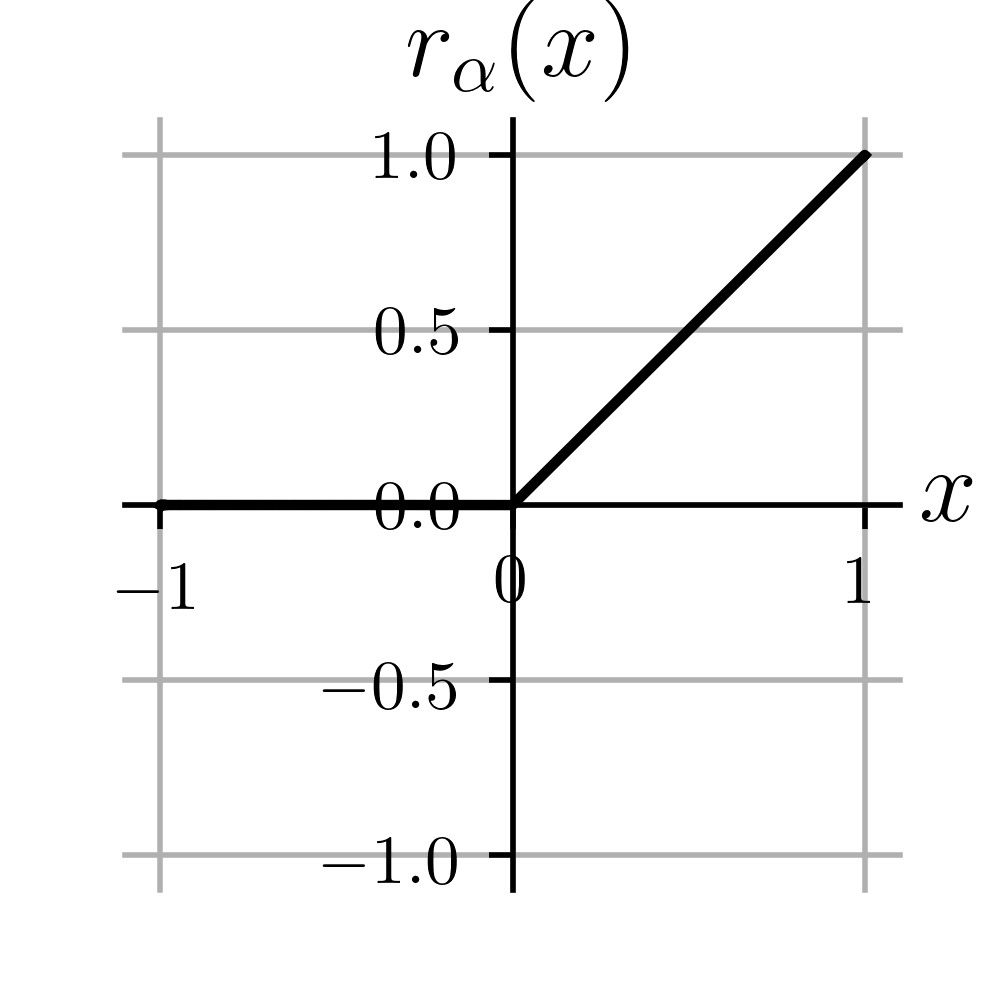}}
\end{minipage}%
\caption{The component polynomials of $p_\alpha (x)$ and the proposed $r_\alpha(x)$ with precision parameter $\alpha=11$. (a): the first component polynomial $p_{\alpha,1}(x)$, (b): the second component polynomial $p_{\alpha,2}(x)$, (c): the third component polynomial $p_{\alpha,3}(x)$, (d): $p_\alpha(x) = p_{\alpha,3} \circ p_{\alpha,2} \circ p_{\alpha,1}(x)$, the composition of $p_{\alpha,1}(x)$, $p_{\alpha,2}(x)$, and $p_{\alpha,3}(x)$, and (e): $r_\alpha (x) = \frac{x+xp_\alpha(x)}{2}$, which is the approximate polynomial of the ReLU function. The degrees of $p_{\alpha,1}(x)$, $p_{\alpha,2}(x)$, and $p_{\alpha,3}(x)$ are 7, 7, and 27, respectively.}
\label{fig:r_alpha}
\end{figure}

\mg{It is well-known that $\ReLU(x) = \frac{x+x\sgn(x)}{2}$. Then, we propose the following approximate polynomial of the ReLU function, $r_{\alpha}(x)= \frac{x+x p_\alpha(x)}{2}$. The specific coefficient values of $r_\alpha(x)$ for several $\alpha$ can be seen in Appendix A. Figure \ref{fig:r_alpha} shows the component polynomials of $p_\alpha (x)$ and the proposed $r_\alpha(x)$ with the precision parameter $\alpha=11$. The following theorem shows that the proposed $r_\alpha(x)$ satisfies the error condition of the ReLU function in equation (\ref{ReLU_error_condition}).}


\begin{theorem}\label{ReLU_thm}
For any $\alpha>0$, $r_\alpha(x)$ satisfies the following inequality:
\begin{equation*}
|r_\alpha(x) - \ReLU(x)| \leq 2^{-\alpha} ~~~ \mbox{for }x \in [-1,1].
\end{equation*}
\end{theorem}
\mg{The proofs of all theorems and a lemma are given in Appendix B.}

\paragraph{Approximation range}
\mg{The value of polynomial $r_\alpha(x)$ is close to the ReLU function value only for $x \in [-1,1]$. However, since some input values of the ReLU function during deep learning are not in $[-1,1]$, the approximation range should be extended from $[-1,1]$ to a larger range $[-B,B]$ for some $B>1$. Thus, we propose a polynomial $\tilde{r}_{\alpha,B}(x) := Br_\alpha(x/B)$ approximating the ReLU function in the range $[-B,B]$. Then, we have $|\tilde{r}_{\alpha,B}(x) - \ReLU(x)|=|B\tilde{r}_{\alpha,B}(x/B) - \ReLU(x)|=|B\tilde{r}_{\alpha,B}(x/B) - B\ReLU(x/B)|=|B(\tilde{r}_{\alpha,B}(x/B) - \ReLU(x/B))|\leq B 2^{-\alpha}$ for $x \in [-B,\:B]$. Large $B$ has the advantage of enlarging input ranges of the ReLU function at the cost of causing a} larger approximation error with the ReLU function. 

\subsection{Precise approximate polynomial of \mg{the} max-pooling function}
The max-pooling function with kernel size $k \times k$ outputs the maximum value among $k^2$ input values.
To implement the max-pooling function for the homomorphically encrypted data, we should find a polynomial that approximates the max-pooling function. 

We define a polynomial $M_{\alpha,n}(x_1,\cdots,x_n)$ that approximates the maximum value function $\max(x_1,\cdots,x_n)$ using $m_\alpha(a,b)$. To reduce the depth consumption of the approximate polynomial of the max-pooling function, we construct $M_{\alpha,n}$ by compositing polynomial $m_\alpha(a,b)$ with a minimal number of compositions. We use the following recursion equation to obtain the approximated maximum value among $x_1,\cdots,x_n$:
\begin{equation}
\label{recursion}
M_{\alpha,n}(x_1,\cdots,x_n) =   \begin{cases}
    x_1, & n=1\\
    m_\alpha(M_{\alpha,k}(x_1,\cdots,x_k),M_{\alpha,k}(x_{k+1},\cdots,x_{2k})), & n=2k \\
    m_\alpha(M_{\alpha,k}(x_1,\cdots,x_k),M_{\alpha,k+1}(x_{k+1},\cdots,x_{2k+1})), & n=2k+1.
  \end{cases}
\end{equation}

The following theorem shows the upper bound of approximation error of  $M_{\alpha,n}(x_1,\cdots,x_n)$.
\begin{theorem}\label{maxpooling_thm}
For given $\alpha>0$ and $n \in \mathbb{N}$, the polynomial $M_{\alpha,n}(x_1,\cdots,x_n)$ obtained from the recursion equation in (\ref{recursion}) satisfies
\begin{equation}
\label{maxpool_ineq}
\begin{aligned}[b]
|M_{\alpha,n}(x_1 , \cdots,x_{n}) -& \max(x_1 , \cdots,x_{n})| \leq 2^{-\alpha} \lceil \log_2 n \rceil \\
&\mathrm{for}\:x_1,\cdots,x_n \in [(\lceil \log_2 n \rceil-1)2^{-\alpha},1-(\lceil \log_2 n \rceil-1)2^{-\alpha}].
\end{aligned}
\end{equation}
\end{theorem}

\paragraph{Approximation range}
To extend the approximation range $[(\lceil \log_2 n \rceil-1)2^{-\alpha},1-(\lceil \log_2 n \rceil-1)2^{-\alpha}]$ to $[-B,B]$ for $B>1$, we use the following polynomial as an approximation polynomial:
\begin{equation*}
    \tilde{M}_{\alpha,n,B}(x_1,\cdots,x_n) = B'\cdot (M_{\alpha,n}(\frac{x_1}{B'}+0.5,\cdots,\frac{x_n}{B'}+0.5)-0.5),
\end{equation*}
where $B' = B/(0.5-(\lceil \log_2 n \rceil-1)2^{-\alpha})$. Then, we have 
\[\left| \tilde{M}_{\alpha,n,B}(x_1,\cdots,x_n) - \max(x_1,\cdots,x_n) \right| \leq B' \cdot 2^{-\alpha}\lceil \log_2 n \rceil,\]
for $x_1,\cdots,x_n \in [-B,B]$, which directly follows from Theorem \ref{maxpooling_thm}. 

\subsection{Comparison between the minimax approximate polynomial and the proposed approximate polynomial for the ReLU and max-pooling functions}

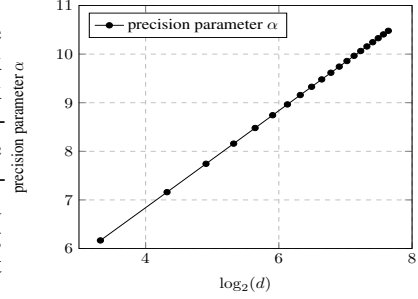
\begin{wrapfigure}{R}{0.4\textwidth}
\vspace{-0cm}
\centering
\resizebox{5cm}{4cm}{%
\hspace*{-1cm}
\begin{tikzpicture}
\begin{axis}[
    xlabel = {$\log_2(d)$},
    ylabel = {precision parameter $\alpha$},
    xmin = 3, xmax = 8,
    ymin = 6, ymax = 11,
    xtick = {2,4,6,8},
    ytick = {6,7,8,9,10,11},
    legend pos = north west,
    ymajorgrids = true,
    xmajorgrids = true,
    grid style = dashed,
]
\addplot[
    mark = *,
    mark size = 1.8pt
]
coordinates {
(3.321928095,6.166431755)
(4.321928095,7.159808652)
(4.906890596,7.743522708)
(5.321928095,8.158121562)
(5.64385619,8.479846327)
(5.906890596,8.742770202)
(6.129283017,8.96509595)
(6.321928095,9.157697743)
(6.491853096,9.327593064)
(6.64385619,9.479574924)
(6.781359714,9.617062736)
(6.906890596,9.742581668)
(7.022367813,9.858049585)
(7.129283017,9.964957408)
(7.22881869,10.06448713)
(7.321928095,10.15759166)
(7.409390936,10.24505046)
(7.491853096,10.32750924)
(7.569855608,10.40550888)
(7.64385619,10.47950702)

};
\addlegendentry{precision parameter $\alpha$}
\end{axis}
\end{tikzpicture}
}
\caption{The achieved precision parameter $\alpha$ according to the logarithm of the polynomial degree (base two).
}
\label{precision_parameter_deg}
\vspace{-0.3cm}
\end{wrapfigure}

One \mg{can think that precise approximation of the ReLU function can be achieved simply by using a high-degree polynomial such as a Taylor polynomial and minimax approximate polynomial \cite{Remez}. In particular, the minimax approximate polynomials have the minimum minimax error for a given degree. Thus, the minimax approximate polynomials are good candidates for the approximate polynomials.} In this subsection, we compare the proposed approximation method using $r_\alpha (x)$ \me{and the method using the minimax approximate polynomial}. 
\mf{For simplicity,} let the approximation range of the ReLU function be $[-1,1]$. 
Then, for a given degree $d$, the minimax approximate polynomial of degree at most $d$ on $[-1,1]$ for $\ReLU(x)$ can be obtained from \mg{the} Remez algorithm \cite{Remez}. \me{The precision parameter $\alpha$ corresponds to the logarithm of the minimax error} using base two. Figure \ref{precision_parameter_deg} presents the \me{corresponding precision parameter} $\alpha$ according to the logarithm of the polynomial degree (base two). \me{Because performing \mg{the} Remez algorithm for a high-degree polynomial has several difficulties\mg{,} such as \mg{using} high real number precision and finding extreme points with high precision}, we obtain the polynomials \me{for} only up to degree $200$. It can be seen that $\log_2 (d)$ and $\alpha$ have an almost exact linear relationship, and we obtain the following regression equation:
\begin{equation}\label{regression}
\md{\alpha = 0.9987 \log_2(d) + 2.8446}.
\end{equation}

Table \ref{comparison_minimax_proposed} compares \mg{the approximation method using the minimax approximate polynomial and that using the proposed approximate polynomial $r_\alpha (x)$}. The required number of non-scalar multiplications for evaluating a polynomial of degree $d$ is $O(\sqrt{2d})$ \cite{Paterson}, and thus, $2^{\Theta(\alpha)}$ non-scalar multiplications are required \me{for the precision parameter} $\alpha$. \mg{The exact number of non-scalar multiplications is obtained using the odd baby-step giant-step algorithm for polynomials with only odd-degree terms and using the optimized baby-step giant-step algorithm for other polynomials \cite{oddbaby}.}
From \md{Table \ref{comparison_minimax_proposed}}, it can be seen that the minimax approximate polynomial of \md{exponentially} high degree is required for a large $\alpha$. This high degree requires many non-scalar multiplications, i.e., large computational time. \me{In addition, using a high-degree polynomial has some numerical issues, such as the requirement of high precision of the real number. Thus, precise approximation using the minimax approximate polynomial is inefficient to use in practice.}


\begin{table}[b]
\caption{Comparison of degrees and number of non-scalar multiplications between the approximation method of using the minimax approximate polynomial and the proposed approximation method. The degrees of the minimax approximate polynomials and the degrees of the component polynomials of $p_\alpha (x)$ are shown. An asterisk(*) indicates that the degree is obtained using the regression equation in (\ref{regression}).}\label{comparison_minimax_proposed}
\centering
\footnotesize{
\begin{tabular}{c|c|c|c|c}
\toprule

\multirow{3}{*}{$\alpha$} & \multicolumn{2}{c|}{minimax}   & \multicolumn{2}{c}{proposed}  \\ 
& \multicolumn{2}{c|}{approx. polynomial}   & \multicolumn{2}{c}{approx. polynomial}  \\ 

\cline{2-5} 
                       & degree & \#mults  & degrees    & \#mults                  \\\hline
6                      & 10                      & 5                          & \{3,7\}                      & 7                          \\
7                      & 20                      & 8                          & \{7,7\}                      & 9                          \\
8                      & 40                      & 12                         & \{7,15\}                     & 12                         \\
9                      & 80                      & 18                         & \{15,15\}                    & 15                         \\
10                     & 150                     & 25                         & \{7,7,13\}                   & 16                         \\
11                     & ~~287*                     & 35                         & \{7,7,27\}                   & 19                         \\
12                     & ~~575*                    & 49                         & \{7,15,27\}                  & 22                         \\
13                     & ~~1151*                   & 70                         & \{15,15,27\}                 & 25                         \\
14                     & ~~2304*                  & 98                         & \{15,27,29\}                 & 28                         \\ 
15 & ~~4612* & 140 & \{29,29,29\} & 30\\

\bottomrule
\end{tabular}
}
\end{table}

\paragraph{\md{Homomorphic} max function}
\me{Considering} $\max(a,b) = \frac{a+b+|a-b|}{2}$, we obtain the minimax approximate polynomial $p(x)$ on $[-1,1]$ for $|x|$. Then, $m(a,b) = \frac{a+b+p(a-b)}{2}$ \me{can be used} as an approximate polynomial of the max function $\max(a,b)$. Since $\ReLU(x) = \frac{x+|x|}{2}$, the approximation of $|x|/2$ and $\ReLU(x)$ \me{are equivalent}. Thus, when the max function is approximated with the minimax approximate polynomial, an exponentially large degree is also required to achieve large $\alpha$.



\section{Theoretical performance analysis of approximate deep learning model}\label{sec:theoretical_analysis}
In this section, we propose an \emph{approximate deep learning model} for homomorphically encrypted input data, replacing the original ReLU and max-pooling functions with the proposed \mg{precise} approximate polynomials $\Tilde{r}_{\alpha,B}(x)$ and $\Tilde{M}_{\alpha,n,B}(x_1 ,\cdots,x_n)$ in the original deep learning model.
In addition, we analytically derive that if the pre-trained parameters of the original deep learning model are available, 
classification would be accurate if we adopt a proper precision parameter $\alpha$. 

In more detail, we estimate the L-infinity norm of the difference between the inference results $\mathcal{F}(\bb{x})$ and $\mathcal{F}^\alpha(\bb{x})$, $\|\mathcal{F}^\alpha(\bb{x})-\mathcal{F}(\bb{x})\|_\infty$.
Here, $\mathcal{F}(\bb{x})$ and $\mathcal{F}^\alpha(\bb{x})$ denote the \mf{inference} results of the original deep learning model $\mathcal{F}$ and the proposed approximate deep learning model $\mathcal{F}^\alpha$ with the precision parameter $\alpha$, \mf{respectively. Here, $\mathcal{F}$ and $\mathcal{F}^\alpha$ share the same pre-trained parameters.}
The main conclusion of this section is that we can reduce the difference $\|\mathcal{F}^\alpha(\bb{x})-\mathcal{F}(\bb{x})\|_\infty$ by increasing $\alpha$, which implies that the proposed model shows almost the same performance as the original pre-trained model.

To estimate $\|\mathcal{F}^\alpha(\bb{x})-\mathcal{F}(\bb{x})\|_\infty$, we attempt to decompose the original deep learning model, analyze each component, and finally re-combine them.
First, we decompose given deep learning model $\mathcal{F}$ as $A_n \circ A_{n-1} \circ \cdots \circ A_0$, \mf{where $A_i$ is called a \emph{block}}. 
\mf{Each block can be a simple function: convolution, batch normalization, activation function, and pooling function. It can} also be a mixed operation that cannot be decomposed into simple operations.
Here, all of the inputs and outputs of block $A_i$'s are considered as \mf{a one-dimensional} vector.
Also, we denote $A_i^\alpha$ as the corresponding approximate polynomial operation of each block $A_i$ with precision parameter $\alpha$.
If block $A_i$ contains only a polynomial operation, $A_i^\alpha$ is the same as $A_i$ since there is nothing to approximate.
Then, the approximate model $\mathcal{F}^{\alpha}$ can be considered as a composition of approximate blocks, ${A}_n^{\alpha} \circ {A}_{n-1}^{\alpha} \circ \cdots \circ {A}_0^{\alpha}$.

If there is only a single approximate block in the approximate deep learning model, the difference in inference results can be easily determined.
However, the model $\mathcal{F}^\alpha$ composed of more than two approximate blocks makes the situation more complicated.
Consider an input vector $\bb{x}$, two consecutive blocks $A_1$ and $A_2$, and their approximation $A_1^\alpha$ and $A_2^\alpha$.
The block $A_1^\alpha$ makes an approximation error which can be measured as $e:=\| A_1^\alpha(\bb{x}) - A_1(\bb{x}) \|_\infty$.
For a given input $\bb{x}$, the magnitude of the error $e$ is only determined by the intrinsic feature of block $A_1$.
However, the output error of the second block is quite different. Note that we can write the second block error as $\| A_2^\alpha(\bb{y}+\bb{e}) - A_2(\bb{y}) \|_\infty$, where $\bb{y} = A_1 (\bb{x})$ and $\bb{e}=A_1^\alpha (\bb{x})-A_1 (\bb{x})$. Considering $\|\bb{e}\|_\infty =e$, both block $A_2$ and the magnitude of $e$ affect the second block error. 
Then, we define a new measure for each block, \mf{\emph{error propagation function}}.

In this study, all inputs $\bb{x}$ in the deep learning models belong to a bounded set for a given dataset. We should set the approximation range $[-B,B]$ for the proposed approximate polynomials for the ReLU and max-pooling functions. We assume that we can choose a sufficiently large $B$ such that the inputs of all blocks fall in $[-B,B]$. That is, for a given deep learning model $\mathcal{F} = A_n \circ A_{n-1} \circ \cdots \circ A_0$ and the corresponding approximate model $\mathcal{F}^\alpha = A_n^\alpha \circ A_{n-1}^\alpha \circ \cdots \circ A_0^\alpha$, we have $\|A_i \circ \cdots \circ A_0 (\bb{x})\|_\infty \leq B$ and $\|A_i^\alpha \circ \cdots \circ A_0^\alpha (\bb{x})\|_\infty \leq B$ for $i, 0 \leq i \leq n-1$ and every input $\bb{x}$. In fact, we confirm through numerical analysis that these are satisfied for an appropriately large $B$ in Section \ref{sec:results}. 
We define error propagation function as follows:

\begin{definition}
For a block $A$, an error propagation function of $A$, $E_A^{\alpha}(\cdot)$, is defined as $E_A^{\alpha}(e) := \sup_{\|\bb{x}+\bb{e}\|_\infty \leq B,\|\bb{e}\|_\infty \leq e} \|{A}^{\alpha}(\bb{x}+\bb{e}) - A(\bb{x})\|_\infty$, where $e$ denotes the magnitude of output error of the previous block of $A$.
\end{definition}

Roughly speaking, block $A^\alpha$ propagates the input error $e$ (or the output error of the previous block) to output error $E_A^\alpha (e)$. 
With this error propagation function, we can estimate the difference between final inference results of $\mathcal{F}$ and $\mathcal{F}^\alpha $ straightforward as in the following theorem.
\begin{theorem}
\label{errorprop}
If the original model $\mathcal{F}$ can be decomposed as $A_n \circ A_{n-1} \circ \cdots \circ A_0$, then we have
\[
\|\mathcal{F}^{\alpha} (\bb{x}) - \mathcal{F}(\bb{x})\|_\infty \leq (E_{A_n} ^{\alpha} \circ E_{A_{n-1}} ^{\alpha} \circ \cdots \circ E_{A_0} ^{\alpha})(0)
\]
for every input $\bb{x}$.
\end{theorem}
This theorem asserts that we can achieve an upper bound of inference error of the whole model by estimating an error propagation function for each block.

We analyze the error propagation function for four types of single blocks: i) linear block, ii) ReLU block, iii) max-pooling block, and iv) softmax block, commonly used in deep neural networks for classification tasks.
From now on, we call these four types of blocks \emph{basic blocks} for convenience. 
Note that convolutional layers, batch normalization layers (for inference), average pooling layers, and fully connected layers are just linear combinations of input vectors and model parameters.
Therefore, it is reasonable to include all these layers in a linear block. 
We can express linear blocks as $A(\bb{x})=\bb{Ax}+\bb{b}$ for some matrix $\bb{A}$ and vector $\bb{b}$. \mf{The following} lemma shows practical upper bounds of error propagation function for basic blocks. \mf{Because the kernel size of the max-pooling function is not greater than ten in most applications, we constrain the kernel size not greater than ten in the following lemma.}
\begin{lemma}
\label{basiclayers}
For the given error $e$, the error propagation functions of four basic blocks are upper bounded as follows: (a) If $A$ is a linear block with $A(\bb{x})=\bb{Ax}+\bb{b}$, then $E_A^\alpha (e) \leq \| \bb{A} \|_\infty e$, where $\| \cdot \|_\infty$ denotes the infinity norm of matrices. 
(b) If $A$ is a ReLU block, then $E_A^\alpha (e) \leq B 2^{-\alpha} + e$. 
(c) If $A$ is a max-pooling block with kernel size $k_0 \leq 10$ and $\alpha \geq 4$, then $E_A^\alpha (e) \leq 10B \lceil \log_2 k_0^2 \rceil 2^{-\alpha}  + e$ . (d) If $A$ is a softmax block, then $E_A^\alpha (e) \leq e/2$.
\end{lemma}
Since the upper bounds of error propagation functions of all basic blocks have been determined, we can obtain an upper bound of the inference error theoretically as in the following theorem.
\begin{theorem}\label{thm:main_thm}
\label{main}
If $\alpha \geq 4$ and a deep learning model $\mathcal{F}$ contains only basic blocks in Lemma \ref{basiclayers}, then there exists a constant $C$ such that $\|\mathcal{F}^{\alpha} (\bb{x}) - \mathcal{F}(\bb{x})\|_\infty \leq C2^{-\alpha}$ for every input $\bb{x}$, where the constant $C$ can be determined independently of the precision parameter $\alpha$. 
\end{theorem}
\mf{
Therefore, the performance of the proposed approximate deep learning model containing only basic blocks can be guaranteed theoretically by Theorem \ref{thm:main_thm} if the original model $\mathcal{F}$ is well trained.}

\paragraph{Practical application of Theorem \ref{thm:main_thm}}
\mg{
Theorem \ref{thm:main_thm} is valid for a deep learning model composed of only basic blocks. VGGNet \cite{simonyan2014very} is a representative example.
However, there are also various deep learning models with non-basic blocks. 
For example,
ResNet \cite{he2016deep} contains \emph{residual learning building block}, which is hard to decompose into several basic blocks. Thus, the generalized version of Theorem \ref{thm:main_thm} is presented in Appendix C, which is also valid for ResNet.
Numerical results in Section \ref{sec:results} also support that the inference result of the original deep learning model and that of the proposed approximate deep learning model are close enough for a practical precision parameter $\alpha$.}

\section{Results of numerical analysis}\label{sec:results}

\mf{In this section, the simulation results are presented for the original and the proposed approximate deep learning models for the plaintext input data.} Overall, our \mf{results} show that using the large precision parameter $\alpha$ obtains high performance in the proposed approximate deep learning model, which \me{confirms} Theorem 4.
\mf{All simulations are proceeded by NVIDIA GeForce RTX 3090 GPU along with AMD Ryzen 9 5950X 16-Core Processor CPU.}
\subsection{Numerical analysis on the CIFAR-10}

\paragraph{Pre-trained backbone models}
We \mf{use the CIFAR-10 \cite{krizhevsky2009learning} as the input data, which has 50k training images and 10k test images with 10 classes.} We use ResNet and VGGNet \mf{with batch normalization} as backbone models of the proposed approximate deep learning model.
\mf{To verify that the proposed approximate deep learning models work well for a various number of layers, we use ResNet with 20, 32, 44, 56, and 110 layers and VGGNet with 11, 13, 16, and 19 layers proposed in \cite{he2016deep} and \cite{simonyan2014very}, respectively. }
In order to obtain pre-trained parameters for each model, we apply the optimizer suggested in \cite{he2016deep} on both ResNet and VGGNet for the CIFAR-10. 

\paragraph{Approximation range} 
The \mf{approximation range $[-B,B]$ should be determined in the proposed approximate deep learning model.} 
\me{We examine the input values of the approximate polynomials of the ReLU and max-pooling functions and determine the value of $B$ by adding some margin to the maximum of the input values.} 
\mf{For the CIFAR-10, we choose $B=50$ for both the approximate polynomials of ReLU and max-pooling functions. However, a larger $B$ is required for retraining with ResNet-110 for $\alpha = 7,8,$ and $9$. Thus, in the case of retraining with ResNet-110, we use $B=360, 195$, and $70$ for $\alpha=7,8,$ and $9$, respectively.
}

\paragraph{Numerical results} 
\begin{table}[]
\centering
\caption{The top-1 accuracy of each approximate deep learning model with respect to $\alpha$ for the plaintext CIFAR-10 dataset. Baseline means the accuracy of the original deep learning model. \mf{The accuracies of the approximate deep learning models with a difference less than 5\% from the baseline accuracy are underlined, and \mg{less} than 1\% are underlined with boldface.}}
{\small
\begin{tabular}{cc|rrrrrrrr}
\toprule
Backbone                                                                & Baseline(\%)                                                                                  & \multicolumn{1}{c}{$\alpha=$7}     & \multicolumn{1}{c}{8}     & \multicolumn{1}{c}{9}     & \multicolumn{1}{c}{10}    & \multicolumn{1}{c}{11}    & \multicolumn{1}{c}{12}    & \multicolumn{1}{c}{13}    & \multicolumn{1}{c}{14}                     \\ \hline
ResNet-20                                                                 & 88.36                                                                                  & 10.00                     & 10.00                     & 10.00                     & 19.96                     & 80.91                     & \uline{86.44}           &\uline{\textbf{88.01}}                    & \uline{\textbf{88.41}}                      \\
ResNet-32                                                                 & 89.38                                                                                  & 10.00                     & 10.00                     & 10.01                     & 36.16                     & \uline{84.92}                     & \uline{\textbf{88.08}}                     & \uline{\textbf{88.99}}                     & \uline{\textbf{89.23}}                     \\
ResNet-44                                                                 & 89.41                                                                                  & 10.00                     & 10.00                     & 10.00                     & 38.35                     & \uline{85.75}                     & \uline{88.31}                     & \uline{\textbf{89.31}}                     & \uline{\textbf{89.37}}                     \\
ResNet-56                                                                 & 89.72                                                                                  & 10.00                     & 10.00                     & 9.94                      & 39.48                     & \uline{86.54}                     & \uline{\textbf{89.18}}                     & \uline{\textbf{89.53}}                     & \uline{\textbf{89.67}}                     \\
ResNet-110                                                                & 89.87                                                                                  & 10.00                     & 10.00                     & 10.04                     & 67.02                     & \uline{86.75}            & \uline{\textbf{89.17}}                     & \uline{\textbf{89.62}}                     & \uline{\textbf{89.84}}                     \\ \hline
VGG-11                                                                    & 90.17                                                                                  & \multicolumn{1}{l}{10.77} & \multicolumn{1}{l}{11.74} & \multicolumn{1}{l}{15.61} & \multicolumn{1}{l}{29.39} & \multicolumn{1}{l}{66.60} & \multicolumn{1}{l}{85.22} & \multicolumn{1}{l}{\uline{88.95}} & \multicolumn{1}{l}{\uline{\textbf{89.91}}} \\
VGG-13                                                                    & 91.90                                                                                  & \multicolumn{1}{l}{11.32} & \multicolumn{1}{l}{11.57} & \multicolumn{1}{l}{14.37} & \multicolumn{1}{l}{36.84} & \multicolumn{1}{l}{77.55} & \multicolumn{1}{l}{\uline{88.72}} & \multicolumn{1}{l}{\uline{90.66}} & \multicolumn{1}{l}{\uline{\textbf{91.37}}} \\
VGG-16                                                                    & 91.99                                                                                  & \multicolumn{1}{l}{10.02} & \multicolumn{1}{l}{10.06} & \multicolumn{1}{l}{10.76} & \multicolumn{1}{l}{32.93} & \multicolumn{1}{l}{78.27} & \multicolumn{1}{l}{\uline{89.69}} & \multicolumn{1}{l}{\uline{\textbf{91.13}}} & \multicolumn{1}{l}{\uline{\textbf{91.87}}} \\
VGG-19                                                                    & 91.75                                                                                  & \multicolumn{1}{l}{10.18} & \multicolumn{1}{l}{13.84} & \multicolumn{1}{l}{18.18} & \multicolumn{1}{l}{43.56} & \multicolumn{1}{l}{80.40} & \multicolumn{1}{l}{\uline{89.42}} & \multicolumn{1}{l}{\uline{\textbf{91.10}}} & \multicolumn{1}{l}{\uline{\textbf{91.61}}}  \\
\bottomrule
\end{tabular}
}
\label{tab:cifar10inf}
\end{table}

First of all, neither the activation functions $x^2$, suggested in \cite{gilad2016cryptonets} and \cite{sealion}, nor $2^{-3} x^2 + 2^{-1} x + 2^{-2}$, suggested in \cite{FasterCryptoNets}, shows good performance in our backbone models.
These activation functions give the inference results around 10.0\% with our pre-trained parameters.
Also, ResNet and VGGNet cannot be trained by the activation functions $x^2$ and $2^{-3} x^2 + 2^{-1} x + 2^{-2}$, since these activation functions cause a more \mf{significant} gradient exploding problem as the original model gets deeper.
Therefore, a polynomial with a small degree is inappropriate to utilize the pre-trained model, \mf{which implies the importance of the proposed precise approximate polynomial for PPML in ResNet and VGGNet.}

The inference results of the proposed approximate deep learning models are given in Table \ref{tab:cifar10inf}.
Table \ref{tab:cifar10inf} shows that if we adopt a large precision parameter $\alpha$, the proposed model becomes more accurate, confirming Theorem \ref{main}.
We achieve the polynomial-operation-only VGG-16 model with 91.87\% accuracy for the CIFAR-10 using $\alpha=14$.
\mf{In addition}, the proposed approximate deep learning model performs almost similarly (1\% difference) to the original deep learning model when $\alpha =12$\mg{--}$14$. 

\begin{table}
\vspace{-0.4cm}
\begin{minipage}{0.6\linewidth}
\centering
\caption{Obtained the top-1 accuracy of each deep learning model with respect to $\alpha$ for the plaintext CIFAR-10 dataset by retraining for ten epochs. \mf{The accuracies with difference less than 5\% from the baseline accuracy are underlined.}}
{\small
\begin{tabular}{cc|rrrr}
\toprule
Backbone                                                                & Baseline                                                                               & \multicolumn{1}{c}{$\alpha=$7}     & \multicolumn{1}{c}{8}     & \multicolumn{1}{c}{9}     & \multicolumn{1}{c}{10}                      \\ \hline
ResNet-20                                                                 & 88.36                                                                 &80.28                 & \uline{84.83}                  & \uline{86.76}                     & \uline{86.61}                                       \\
ResNet-32                                                                 & 89.38                                                                                  & 80.09                    & \uline{86.08}                    & \uline{87.76}                     & \uline{87.85}                                      \\
ResNet-44                                                                 & 89.41                                                                                  & 76.13                     & 82.93                    & \uline{85.61}                     & \uline{88.06}                                     \\
ResNet-56                                                                 & 89.72                                                                                  & 70.90                     & 77.59                     & 85.06                      & \uline{87.96}                                       \\
ResNet-110                                                                & 89.87                                                                                  & 50.47                     & 31.24                     & 79.93                    & \uline{88.28}                                       \\ \hline
VGG-11                                                                    & 90.17                                                                                  & \multicolumn{1}{l}{66.54} & \multicolumn{1}{l}{78.57} & \multicolumn{1}{l}{\uline{86.56}} & \multicolumn{1}{l}{\uline{88.04}}  \\
VGG-13                                                                    & 91.90                                                                                  & \multicolumn{1}{l}{73.09} & \multicolumn{1}{l}{83.67} & \multicolumn{1}{l}{\uline{88.97}} & \multicolumn{1}{l}{\uline{90.74}} \\
VGG-16                                                                    & 91.99                                                                                  & \multicolumn{1}{l}{73.75} & \multicolumn{1}{l}{84.12} & \multicolumn{1}{l}{\uline{89.37}} & \multicolumn{1}{l}{\uline{90.65}} \\
VGG-19                                                                    & 91.75                                                                                  & \multicolumn{1}{l}{75.54} & \multicolumn{1}{l}{83.77} & \multicolumn{1}{l}{\uline{88.90}} & \multicolumn{1}{l}{\uline{90.21}}  \\
\bottomrule
\end{tabular}
}
\label{tab:retraining}
\end{minipage}
\begin{minipage}{0.4\linewidth}
    \caption{The top-1 accuracy of each approximate deep learning model for the plaintext ImageNet dataset with $\alpha=14$. \mf{The accuracies are underlined and boldfaced in the same manner as in Table \ref{tab:cifar10inf}.}}
    \label{tab:infacc_imagenet}
    \centering
    {\small
\begin{tabular}{ccc}
\toprule
Backbone& Baseline & Proposed \\ \hline
Inception-v3 \cite{szegedy2016rethinking}                                            & 69.54                                                              & \uline{\textbf{69.23}}                                                                     \\
GoogLeNet \cite{szegedy2015going}                                                & 68.13                                                              &\uline{69.76}                                                                             \\
VGG-19 \cite{simonyan2014very}                                                   & 74.22                                                              &      \uline{73.46}           \\
ResNet-152 \cite{he2016deep}                                               & 78.31                                                              &     \uline{77.52}                 \\
\bottomrule
\end{tabular}
}
\end{minipage}
\vspace{-0.6cm}
\end{table}

\mf{Moreover, if retraining is allowed for a few epochs, we also obtain} high performance in the proposed approximate deep learning model with a \mf{smaller} precision parameter $\alpha$.
We retrain the proposed model using the same optimizer for pre-training but only for \mf{ten} epochs and drop learning rate with multiplicative factor $0.1$ at the 5th epoch. 
The results of retraining with precision parameter $\alpha=7 $\mg{--}$10$ are summarized in Table \ref{tab:retraining}.
\mf{Comparing Tables \ref{tab:cifar10inf} and \ref{tab:retraining} shows that the required precision parameter $\alpha$ for high classification accuracy is reduced if we allow retraining for ten epochs.}

\subsection{Numerical analysis on the ImageNet}
\mf{We also analyze the classification accuracy of the proposed approximate deep learning model for the ImageNet dataset \cite{russakovsky2015imagenet}, which has 1.28M training images and 50k validation images with 1000 classes.}
\mf{We simulate various models whose pre-trained parameters are available for inference of the ImageNet. }
For the approximation range $[-B,B]$, we choose $B=400$ for the ReLU function and $B=100$ for the max-pooling function. 
The models that we use and their inference results are summarized in Table \ref{tab:infacc_imagenet}.
Without any training, we easily achieve a polynomial-operation-only model with an accuracy of 77.52\% for the ImageNet.
\mf{From this result, it can be seen that even if the dataset is large, the proposed model shows high performance, similar to the original model.}





\subsection{Discussion of classification accuracy performance on FHE}\label{sec:discussion_fhe}
In this study, deep learning was performed for plaintext input data, not for the homomorphically encrypted input data. \mg{However, since polynomials replaced all non-arithmetic operations, the proposed approximate deep learning model can also efficiently be performed in FHE, and the classification accuracy would be identical. 
In particular, the proposed approximate deep learning model will work well in the RNS-CKKS scheme \cite{RNS-CKKS} because the computational time and precision of RNS-CKKS scheme operations have been improved significantly in recent years \cite{GPUCKKS,lee2020optimal, variance, scale_invariant}.}


\section{Conclusion}
We proposed the polynomials that precisely approximate the ReLU and max-pooling functions using a composition of minimax approximate polynomials of small degrees for PPML using FHE. We theoretically showed that if the precision parameter $\alpha$ is large enough, the difference between the inference result of the original and the proposed approximate models becomes small enough. For the CIFAR-10 dataset, the \mg{proposed model approximating} ResNet or VGGNet had classification accuracy with only 1\% error from the original deep learning model when we used the precision parameter $\alpha$ from 12 to 14. Furthermore, for the first time for PPML using word-wise HE, we achieved 77.52\% of test accuracy for ImageNet classification, which was close to the original accuracy of 78.31\% for the precision parameter $\alpha = 14$.

\newpage

\section*{Appendix}

\appendix
\section{Coefficients of the approximate polynomial of the sign function, $p_\alpha(x)$}
In this section, for the approximate polynomial of the sign function, $p_\alpha(x)$, the coefficients of the component polynomials $p_{\alpha,1}(x), \cdots , p_{\alpha,k}(x)$ using the Remez algorithm \cite{Remez} are shown for the precision parameter $\alpha=7,8,\cdots,14$. Using the coefficients in the table below, we can obtain the proposed approximate polynomial of the ReLU function using the following equation:
\[
r_\alpha(x) = \frac{1}{2}(x+x p_\alpha(x)) = \frac{1}{2}(x+x(p_{\alpha,k} \circ \cdots \circ p_{\alpha,1})(x)).
\]

Note that $p_{\alpha,i}(x)$ is an odd function theoretically, and the obtained coefficients of even-degree terms of $p_{\alpha,i}(x)$ by using the Remez algorithm \cite{Remez} are almost zero. Even if we delete the even-degree terms of $p_{\alpha,i}(x)$ and obtain $r_\alpha(x)$ using the modified $p_{\alpha,i}(x)$, the approximate deep learning model using the modified $r_\alpha(x)$ has almost the same inference results.

\begin{longtable}{c|cc|S[table-format=-1.14e-2]}
\toprule
$\alpha$ & $i$ & $m$ & {\centering \text{Coefficients of} $x^m$ \text{of} $p_{\alpha,i}(x)$ }\\ \hline
7        & 1   & 0   & 3.60471572275560e-36                              \\
         &     & 1   & 7.30445164958251e+00                              \\
         &     & 2   & -5.05471704202722e-35                             \\
         &     & 3   & -3.46825871108659e+01                             \\
         &     & 4   & 1.16564665409095e-34                              \\
         &     & 5   & 5.98596518298826e+01                              \\
         &     & 6   & -6.54298492839531e-35                             \\
         &     & 7   & -3.18755225906466e+01                             \\ 
         &  2   & 0   & -9.46491402344260e-49                             \\
         &     & 1   & 2.40085652217597e+00                              \\
         &     & 2   & 6.41744632725342e-48                              \\
         &     & 3   & -2.63125454261783e+00                             \\
         &     & 4   & -7.25338564676814e-48                             \\
         &     & 5   & 1.54912674773593e+00                              \\
         &     & 6   & 2.06916466421812e-48                              \\
         &     & 7   & -3.31172956504304e-01                            \\ \hline
      8   & 1 & 0 & -5.30489756589578E-48 \\
 &  &1  & 8.83133072022416E+00  \\
 &  & 2 & 2.15841891006552E-46  \\
 &  & 3 & -4.64575039895512E+01 \\
 &  & 4 & -6.58937888826136E-46 \\
 &  & 5 & 8.30282234720408E+01  \\
 &  & 6 & 4.43205802152239E-46  \\
 &  & 7 & -4.49928477828070E+01 \\
 &2  & 0 & -3.39457286447112E-32 \\
 &  & 1 & 3.94881885083263E+00  \\
 &  & 2 & 8.77744308285903E-31  \\
 &  & 3 & -1.29103010992282E+01 \\
 &  & 4 & -3.73356852706615E-30 \\
 &  & 5 & 2.80865362174658E+01  \\
 &  & 6 & 5.59273808588447E-30  \\
 &  & 7 & -3.55969148965137E+01 \\
 &  & 8 & -3.36963375307073E-30 \\
 &  & 9 & 2.65159370881337E+01  \\
 &  & 10 & 5.36813679148778E-31  \\
 &  & 11 & -1.14184889368449E+01 \\
 &  & 12 & 1.91081017684427E-31  \\
 &  & 13 & 2.62558443881334E+00  \\
 &  & 14 & -5.50686982942230E-32 \\
 &  & 15 & -2.49172299998642E-01 \\ \hline
 9 & 1 & 0  & 3.85169741234183E-44  \\
  &   & 1  & 1.80966285718807E+01  \\
  &   & 2  & -4.59730416916377E-42 \\
  &   & 3  & -4.34038703274886E+02 \\
  &   & 4  & 7.96299160375690E-41  \\
  &   & 5  & 4.15497103545696E+03  \\
  &   & 6  & -5.28977110396316E-40 \\
  &   & 7  & -1.86846943613149E+04 \\
  &   & 8  & 1.67219551148917E-39  \\
  &   & 9  & 4.41657177889329E+04  \\
  &   & 10 & -2.69777424798506E-39 \\
  &   & 11 & -5.65527928983401E+04 \\
  &   & 12 & 2.14124591383569E-39  \\
  &   & 13 & 3.71156122725781E+04  \\
  &   & 14 & -6.61722455927198E-40 \\
  &   & 15 & -9.78241933892781E+03 \\
  & 2 & 0  & -1.04501074063854E-46 \\
  &   & 1  & 3.79753323360856E+00  \\
  &   & 2  & 4.22842209818016E-45  \\
  &   & 3  & -1.17718157771192E+01 \\
  &   & 4  & -2.25571113936639E-44 \\
  &   & 5  & 2.49771086678346E+01  \\
  &   & 6  & 4.42462875106862E-44  \\
  &   & 7  & -3.15238841603993E+01 \\
  &   & 8  & -4.13554194411645E-44 \\
  &   & 9  & 2.37294863126722E+01  \\
  &   & 10 & 2.00060158783094E-44  \\
  &   & 11 & -1.04331800195923E+01 \\
  &   & 12 & -4.86041132712796E-45 \\
  &   & 13 & 2.46743976260838E+00  \\
  &   & 14 & 4.71256214052049E-46  \\
  &   & 15 & -2.42130100247617E-01 \\ \hline
  10 & 1 & 0  & -1.68048812248597E-47 \\
  &   & 1  & 1.08541842577442E+01  \\
  &   & 2  & 5.19213405604261E-46  \\
  &   & 3  & -6.22833925211098E+01 \\
  &   & 4  & -1.67358715007438E-45 \\
  &   & 5  & 1.14369227820443E+02  \\
  &   & 6  & 1.15437076692363E-45  \\
  &   & 7  & -6.28023496973074E+01 \\
  & 2 & 0  & 7.86253562483970E-39  \\
  &   & 1  & 4.13976170985111E+00  \\
  &   & 2  & -7.18241741649940E-38 \\
  &   & 3  & -5.84997640211679E+00 \\
  &   & 4  & 5.17878634442782E-38  \\
  &   & 5  & 2.94376255659280E+00  \\
  &   & 6  & -9.33059743960049E-39 \\
  &   & 7  & -4.54530437460152E-01 \\
  & 3 & 0  & 3.75374153583292E-39  \\
  &   & 1  & 3.29956739043733E+00  \\
  &   & 2  & -1.04537140020889E-37 \\
  &   & 3  & -7.84227260291355E+00 \\
  &   & 4  & 4.18647895984231E-37  \\
  &   & 5  & 1.28907764115564E+01  \\
  &   & 6  & -6.09510159540855E-37 \\
  &   & 7  & -1.24917112584486E+01 \\
  &   & 8  & 4.05475441247124E-37  \\
  &   & 9  & 6.94167991428074E+00  \\
  &   & 10 & -1.26770087815848E-37 \\
  &   & 11 & -2.04298067399942E+00 \\
  &   & 12 & 1.52452197400636E-38  \\
  &   & 13 & 2.46407138926031E-01  \\ \hline
11 & 1 & 0  & -7.57739739977406E-31 \\
  &   & 1  & 1.12590667402954E+01  \\
  &   & 2  & 3.50298736439109E-29  \\
  &   & 3  & -6.54692933329973E+01 \\
  &   & 4  & -1.10939529976921E-28 \\
  &   & 5  & 1.20694634277757E+02  \\
  &   & 6  & 7.59102357594000E-29  \\
  &   & 7  & -6.64019695377825E+01 \\
  & 2 & 0  & 6.70745934852179E-49  \\
  &   & 1  & 4.70477624210883E+00  \\
  &   & 2  & -3.56089594615542E-48 \\
  &   & 3  & -6.79884851596681E+00 \\
  &   & 4  & 1.43595646069924E-48  \\
  &   & 5  & 3.31525104382873E+00  \\
  &   & 6  & -1.49853421385792E-49 \\
  &   & 7  & -4.89362936859897E-01 \\
  & 3 & 0  & -3.72466666381643E-46 \\
  &   & 1  & 5.36334257654953E+00  \\
  &   & 2  & 4.38732568853777E-44  \\
  &   & 3  & -3.55169555441962E+01 \\
  &   & 4  & -7.27558627095135E-43 \\
  &   & 5  & 1.77807304115644E+02  \\
  &   & 6  & 4.67563353147443E-42  \\
  &   & 7  & -5.92297395415024E+02 \\
  &   & 8  & -1.55919114948318E-41 \\
  &   & 9  & 1.34891691889362E+03  \\
  &   & 10 & 3.10198751789898E-41  \\
  &   & 11 & -2.15876445084938E+03 \\
  &   & 12 & -3.97952274644606E-41 \\
  &   & 13 & 2.47365685586918E+03  \\
  &   & 14 & 3.44258367173957E-41  \\
  &   & 15 & -2.04913542536248E+03 \\
  &   & 16 & -2.05447361791160E-41 \\
  &   & 17 & 1.22739317090559E+03  \\
  &   & 18 & 8.49694447438626E-42  \\
  &   & 19 & -5.25826175134002E+02 \\
  &   & 20 & -2.39856365003446E-42 \\
  &   & 21 & 1.56930558712840E+02  \\
  &   & 22 & 4.42223468597695E-43  \\
  &   & 23 & -3.09658595645912E+01 \\
  &   & 24 & -4.81331429819361E-44 \\
  &   & 25 & 3.62894000814968E+00  \\
  &   & 26 & 2.35181370578811E-45  \\
  &   & 27 & -1.91160283749939E-01 \\ \hline
12 & 1 & 0  & 3.28671253798158E-45  \\
  &   & 1  & 1.15523042357223E+01  \\
  &   & 2  & -2.92342552286817E-43 \\
  &   & 3  & -6.77794513440968E+01 \\
  &   & 4  & 9.33659553243619E-43  \\
  &   & 5  & 1.25283740404562E+02  \\
  &   & 6  & -6.41319512076187E-43 \\
  &   & 7  & -6.90142908232934E+01 \\
  & 2 & 0  & 6.41087388948633E-46  \\
  &   & 1  & 9.65167636181626E+00  \\
  &   & 2  & -1.22822329506037E-43 \\
  &   & 3  & -6.16939174538469E+01 \\
  &   & 4  & 6.20624566340835E-43  \\
  &   & 5  & 1.55170351652298E+02  \\
  &   & 6  & -9.96218491919333E-43 \\
  &   & 7  & -1.82697582383214E+02 \\
  &   & 8  & 7.27748968270610E-43  \\
  &   & 9  & 1.12910726525406E+02  \\
  &   & 10 & -2.69134924845614E-43 \\
  &   & 11 & -3.77752411770263E+01 \\
  &   & 12 & 4.93235088742835E-44  \\
  &   & 13 & 6.47503909732344E+00  \\
  &   & 14 & -3.56876826458906E-45 \\
  &   & 15 & -4.45613365723361E-01 \\
  & 3 & 0  & 4.77710576312791E-47  \\
  &   & 1  & 5.25888355571745E+00  \\
  &   & 2  & -2.94571921438375E-45 \\
  &   & 3  & -3.37233593794284E+01 \\
  &   & 4  & 4.43279401132879E-44  \\
  &   & 5  & 1.64983085013457E+02  \\
  &   & 6  & -2.76863985519552E-43 \\
  &   & 7  & -5.41408891406992E+02 \\
  &   & 8  & 9.15181002994263E-43  \\
  &   & 9  & 1.22296207997963E+03  \\
  &   & 10 & -1.82071128047940E-42 \\
  &   & 11 & -1.95201910566479E+03 \\
  &   & 12 & 2.34852758455781E-42  \\
  &   & 13 & 2.24084021378300E+03  \\
  &   & 14 & -2.05169300503205E-42 \\
  &   & 15 & -1.86634916983170E+03 \\
  &   & 16 & 1.24157528220800E-42  \\
  &   & 17 & 1.12722117843121E+03  \\
  &   & 18 & -5.22850779777304E-43 \\
  &   & 19 & -4.88070474638380E+02 \\
  &   & 20 & 1.50923249713814E-43  \\
  &   & 21 & 1.47497846308920E+02  \\
  &   & 22 & -2.85763045364643E-44 \\
  &   & 23 & -2.95171048879526E+01 \\
  &   & 24 & 3.20768549654760E-45  \\
  &   & 25 & 3.51269520930994E+00  \\
  &   & 26 & -1.62263985493395E-46 \\
  &   & 27 & -1.88101836557879E-01 \\ \hline
13 & 1 & 0  & 1.34595769293910E-33  \\
  &   & 1  & 2.45589415425004E+01  \\
  &   & 2  & 4.85095667238242E-32  \\
  &   & 3  & -6.69660449716894E+02 \\
  &   & 4  & -2.44541235853840E-30 \\
  &   & 5  & 6.67299848301339E+03  \\
  &   & 6  & 1.86874811944640E-29  \\
  &   & 7  & -3.06036656163898E+04 \\
  &   & 8  & -5.76227817577242E-29 \\
  &   & 9  & 7.31884032987787E+04  \\
  &   & 10 & 8.53680673009259E-29  \\
  &   & 11 & -9.44433217050084E+04 \\
  &   & 12 & -6.02701474694667E-29 \\
  &   & 13 & 6.23254094212546E+04  \\
  &   & 14 & 1.62342843661940E-29  \\
  &   & 15 & -1.64946744117805E+04 \\
  & 2 & 0  & 1.53261588585630E-47  \\
  &   & 1  & 9.35625636035439E+00  \\
  &   & 2  & -3.68972123048249E-46 \\
  &   & 3  & -5.91638963933626E+01 \\
  &   & 4  & 1.74254399703303E-45  \\
  &   & 5  & 1.48860930626448E+02  \\
  &   & 6  & -3.20672110002213E-45 \\
  &   & 7  & -1.75812874878582E+02 \\
  &   & 8  & 2.79115738948645E-45  \\
  &   & 9  & 1.09111299685955E+02  \\
  &   & 10 & -1.22590309306100E-45 \\
  &   & 11 & -3.66768839978755E+01 \\
  &   & 12 & 2.62189142557962E-46  \\
  &   & 13 & 6.31846290311294E+00  \\
  &   & 14 & -2.16662326421275E-47 \\
  &   & 15 & -4.37113415082177E-01 \\
  & 3 & 0  & 6.43551938319983E-48  \\
  &   & 1  & 5.07813569758861E+00  \\
  &   & 2  & 8.12601038855762E-46  \\
  &   & 3  & -3.07329918137186E+01 \\
  &   & 4  & -1.60198474678427E-44 \\
  &   & 5  & 1.44109746812809E+02  \\
  &   & 6  & 1.07463154460511E-43  \\
  &   & 7  & -4.59661688826142E+02 \\
  &   & 8  & -3.63448723044512E-43 \\
  &   & 9  & 1.02152064470459E+03  \\
  &   & 10 & 7.25207125369784E-43  \\
  &   & 11 & -1.62056256708877E+03 \\
  &   & 12 & -9.27306397853655E-43 \\
  &   & 13 & 1.86467646416570E+03  \\
  &   & 14 & 7.95843097354065E-43  \\
  &   & 15 & -1.56749300877143E+03 \\
  &   & 16 & -4.69190103147527E-43 \\
  &   & 17 & 9.60970309093422E+02  \\
  &   & 18 & 1.90863349654016E-43  \\
  &   & 19 & -4.24326161871646E+02 \\
  &   & 20 & -5.27439678020696E-44 \\
  &   & 21 & 1.31278509256003E+02  \\
  &   & 22 & 9.47044937974786E-45  \\
  &   & 23 & -2.69812576626115E+01 \\
  &   & 24 & -9.98181561763750E-46 \\
  &   & 25 & 3.30651387315565E+00  \\
  &   & 26 & 4.69390466192199E-47  \\
  &   & 27 & -1.82742944627533E-01 \\ \hline
14 & 1 & 0  & -3.38572283433492E-47 \\
  &   & 1  & 2.49052143193754E+01  \\
  &   & 2  & 7.67064296707865E-45  \\
  &   & 3  & -6.82383057582430E+02 \\
  &   & 4  & -1.33318527258859E-43 \\
  &   & 5  & 6.80942845390599E+03  \\
  &   & 6  & 9.19464568002043E-43  \\
  &   & 7  & -3.12507100017105E+04 \\
  &   & 8  & -3.02547883089949E-42 \\
  &   & 9  & 7.47659388363757E+04  \\
  &   & 10 & 5.02426027571770E-42  \\
  &   & 11 & -9.65046838475839E+04 \\
  &   & 12 & -4.05931240321443E-42 \\
  &   & 13 & 6.36977923778246E+04  \\
  &   & 14 & 1.26671427827897E-42  \\
  &   & 15 & -1.68602621347190E+04 \\
  & 2 & 0  & -9.27991756967991E-46 \\
  &   & 1  & 1.68285511926011E+01  \\
  &   & 2  & 8.32408114686671E-44  \\
  &   & 3  & -3.39811750495659E+02 \\
  &   & 4  & -1.27756566625811E-42 \\
  &   & 5  & 2.79069998793847E+03  \\
  &   & 6  & 7.70152836729131E-42  \\
  &   & 7  & -1.13514151573790E+04 \\
  &   & 8  & -2.41159918805990E-41 \\
  &   & 9  & 2.66230010283745E+04  \\
  &   & 10 & 4.48807056213874E-41  \\
  &   & 11 & -3.93840328661975E+04 \\
  &   & 12 & -5.34821622972202E-41 \\
  &   & 13 & 3.87884230348060E+04  \\
  &   & 14 & 4.25722502798559E-41  \\
  &   & 15 & -2.62395303844988E+04 \\
  &   & 16 & -2.31146624263347E-41 \\
  &   & 17 & 1.23656207016532E+04  \\
  &   & 18 & 8.58571463533718E-42  \\
  &   & 19 & -4.05336460089999E+03 \\
  &   & 20 & -2.14564940301255E-42 \\
  &   & 21 & 9.06042880951087E+02  \\
  &   & 22 & 3.44803367899992E-43  \\
  &   & 23 & -1.31687649208288E+02 \\
  &   & 24 & -3.21717059336602E-44 \\
  &   & 25 & 1.12176079033623E+01  \\
  &   & 26 & 1.32425600403443E-45  \\
  &   & 27 & -4.24938020467471E-01 \\
  & 3 & 0  & 6.72874968716530E-48  \\
  &   & 1  & 5.31755497689391E+00  \\
  &   & 2  & 5.68199275801086E-46  \\
  &   & 3  & -3.54371531531577E+01 \\
  &   & 4  & -1.35187813155454E-44 \\
  &   & 5  & 1.84122441329140E+02  \\
  &   & 6  & 1.05531766289589E-43  \\
  &   & 7  & -6.55386830146253E+02 \\
  &   & 8  & -4.14266518871760E-43 \\
  &   & 9  & 1.63878335428060E+03  \\
  &   & 10 & 9.63097361166316E-43  \\
  &   & 11 & -2.95386237048226E+03 \\
  &   & 12 & -1.44556688409360E-42 \\
  &   & 13 & 3.90806423362418E+03  \\
  &   & 14 & 1.47265013864485E-42  \\
  &   & 15 & -3.83496739165131E+03 \\
  &   & 16 & -1.04728251169615E-42 \\
  &   & 17 & 2.79960654766517E+03  \\
  &   & 18 & 5.26108728786276E-43  \\
  &   & 19 & -1.51286231886692E+03 \\
  &   & 20 & -1.86083902222546E-43 \\
  &   & 21 & 5.96160139340009E+02  \\
  &   & 22 & 4.53644110199468E-44  \\
  &   & 23 & -1.66321739302958E+02 \\
  &   & 24 & -7.25782287655313E-45 \\
  &   & 25 & 3.10988369739884E+01  \\
  &   & 26 & 6.85800520634485E-46  \\
  &   & 27 & -3.49349374506190E+00 \\
  &   & 28 & -2.89849811206637E-47 \\
  &   & 29 & 1.78142156956495E-01  \\
  
\bottomrule
\end{longtable}

\section{Proofs of theorems and a lemma}

\paragraph{Proof of Theorem \ref{ReLU_thm}.}
\begin{proof}
From the definition, $p_\alpha (x)$ satisfies the following inequality:
\begin{align*}
|m_\alpha(a,b) - \max(a,b)| = |\frac{(a+b)+(a-b)p_\alpha (a-b)}{2} - \max(a,b) | \leq 2^{-\alpha},
\end{align*}
for $a,b \in [0,1]$. Then, for $x \in [0,1]$, we have $|r_\alpha (x) - \ReLU(x)| = |\frac{x+x p_\alpha(x)}{2} - x| = |m_\alpha(x,0) - \max(x,0)| \leq 2^{-\alpha}$. In addition, for $x \in [-1,0]$, we have $|r_\alpha (x) - \ReLU(x)| = |\frac{x+x p_\alpha(x)}{2}| = |\frac{-x+x p_\alpha(x)}{2} + x| = |m_\alpha(0,-x) - \max(0,-x)| \leq 2^{-\alpha}$. Thus, we have $|r_\alpha(x) - \ReLU(x)| \leq 2^{-\alpha} \mbox{~for }x \in [-1,1]$.
\end{proof}

\paragraph{Proof of Theorem \ref{maxpooling_thm}.}
\begin{proof}

To prove this theorem, we require two lemmas.\\
(Lemma (a)) For $n \in \mathbb{N}$, let $A$ and $B$ satisfy 
\begin{equation}\notag
(\lceil \log_2 n \rceil -1) 2^{-\alpha} \leq A < B \leq 1 - (\lceil \log_2 n \rceil -1) 2^{-\alpha}.
\end{equation}
Then, for $x_1, x_2, \cdots, x_n \in [A,B]$, the following is satisfied:
\begin{equation}\notag
M_{\alpha,n}(x_1, \cdots, x_n) \in [A-\lceil \log_2 n \rceil 2^{-\alpha}, B+\lceil \log_2 n \rceil 2^{-\alpha}].
\end{equation}

(Proof) We will use mathematical induction to show that Lemma (a) holds for all $n \in N$. For $n=1$, it is trivial because $M_{\alpha,n}(x_1, \cdots, x_n) \in [A-\lceil \log_2 n \rceil 2^{-\alpha}, B+\lceil \log_2 n \rceil 2^{-\alpha}]$ if and only if $x_1 \in [A,B]$.

For $n=2$, we have $0 \leq A < B \leq 1$. We have to show that $M_{\alpha,2}(x_1,x_2) = m_\alpha (x_1, x_2) \in [A-2^{-\alpha},B+2^{-\alpha}]$ for $x_1, x_2 \in [A,B]$. We note that 
\begin{equation}\label{m_alpha_condition}
|m_\alpha (a,b) - \max(a,b)|\leq 2^{-\alpha}~~ \mbox{for}~ a,b \in [0,1].    
\end{equation}
Because $|m_\alpha (x_1,x_2) - \max(x_1,x_2)|\leq 2^{-\alpha}$, we have 
\begin{align*}
-2^{-\alpha}+A &\leq m_\alpha (x_1,x_2) - \max(x_1,x_2) + A\\
&\leq m_\alpha (x_1,x_2) - \max(x_1,x_2) + \max(x_1,x_2) \\
&= m_\alpha (x_1,x_2).
\end{align*}
 Also, we have 
\begin{align*}
m_\alpha (x_1,x_2) &= m_\alpha (x_1,x_2) - \max(x_1,x_2) + \max(x_1,x_2) \\
&\leq 2^{-\alpha}+\max(x_1,x_2) \\
&\leq 2^{-\alpha} +B.
\end{align*}
Thus, Lemma (a) holds for $n=2$. Now, we assume that Lemma (a) holds for $n$, $1 \leq n \leq m-1$ for some $m \geq 3$. It is enough to show that Lemma (a) also holds for $n=m$. \\\\
(i) $m=2k$\\
We have
\begin{equation}\label{ABrange}
(\lceil \log_2 2k \rceil -1) 2^{-\alpha} \leq A < B \leq 1- (\lceil \log_2 2k \rceil -1) 2^{-\alpha},
\end{equation}
which is equivalent to $(\lceil \log_2 k \rceil) 2^{-\alpha} \leq A < B \leq 1- (\lceil \log_2 k \rceil) 2^{-\alpha}.$ Then, we have to show that 
\begin{equation}\notag
M_{\alpha,2k}(x_1, \cdots, x_{2k}) \in [A-\lceil \log_2 2k \rceil 2^{-\alpha}, B+\lceil \log_2 2k \rceil 2^{-\alpha}],
\end{equation}
for $x_1, \cdots , x_{2k} \in [A,B]$.
Because Lemma (a) holds for $n=k$ by the intermediate induction assumption, we have
\begin{align*}
&M_{\alpha,k} (x_1,\cdots,x_k) \in [A-\lceil \log_2 k \rceil 2^{-\alpha},B+\lceil \log_2 k \rceil 2^{-\alpha}] ~~\mbox{and}\\
&M_{\alpha,k} (x_{k+1},\cdots,x_{2k}) \in [A-\lceil \log_2 k \rceil 2^{-\alpha},B+\lceil \log_2 k \rceil 2^{-\alpha}].
\end{align*}

From the inequality in (\ref{ABrange}), we have $0 \leq A-\lceil \log_2 k \rceil 2^{-\alpha}$ and $B+\lceil \log_2 k \rceil 2^{-\alpha} \leq 1$. Thus, we have 
\[
M_{\alpha,k} (x_1,\cdots,x_k), M_{\alpha,k} (x_{k+1},\cdots,x_{2k}) \in [A-\lceil \log_2 k \rceil 2^{-\alpha},B+\lceil \log_2 k \rceil 2^{-\alpha}] \subseteq [0,1].
\]
Then, from Lemma (a) for $n=2$, we have
\begin{align*}
M_{\alpha,2k}(x_1, \cdots, x_{2k}) &= m_\alpha (M_{\alpha,k} (x_1,\cdots,x_k), M_{\alpha,k} (x_{k+1},\cdots,x_{2k})) \\
&\in [A-\lceil \log_2 k \rceil 2^{-\alpha} - 2^{-\alpha},B+\lceil \log_2 k \rceil 2^{-\alpha} + 2^{-\alpha}] \\
&= [A-\lceil \log_2 2k \rceil 2^{-\alpha},B+\lceil \log_2 2k \rceil 2^{-\alpha}]
\end{align*}
Thus, Lemma (a) holds for $n=m=2k$.\\

(ii) $m=2k+1$\\
We have
\begin{equation}\notag
(\lceil \log_2 (2k+1) \rceil -1) 2^{-\alpha} \leq A < B \leq 1- (\lceil \log_2 (2k+1) \rceil -1) 2^{-\alpha}.
\end{equation}
This is equivalent to 
\begin{equation}\label{ABrange2}
    (\lceil \log_2 (k+1) \rceil) 2^{-\alpha} \leq A < B \leq 1- (\lceil \log_2 (k+1) \rceil) 2^{-\alpha}
\end{equation}
because $\lceil \log_2 (2k+1) \rceil = \lceil \log_2 (2k+2) \rceil$ for every integer $k\geq1$. Then, we have to show that 
\begin{equation}\notag
M_{\alpha,2k+1}(x_1, \cdots, x_{2k+1}) \in [A-\lceil \log_2 (2k+1) \rceil 2^{-\alpha}, B+\lceil \log_2 (2k+1) \rceil 2^{-\alpha}],
\end{equation}
for $x_1, \cdots , x_{2k+1} \in [A,B]$.
Because Lemma (a) holds for $n=k$ and $n=k+1$ by the intermediate induction assumption, we have
\begin{align*}
&M_{\alpha,k} (x_1,\cdots,x_k) \in [A-\lceil \log_2 k \rceil 2^{-\alpha},B+\lceil \log_2 k \rceil 2^{-\alpha}] ~~\mbox{and}\\
&M_{\alpha,k+1} (x_{k+1},\cdots,x_{2k+1}) \in [A-\lceil \log_2 (k+1) \rceil 2^{-\alpha},B+\lceil \log_2 (k+1) \rceil 2^{-\alpha}].
\end{align*}

From the inequality in (\ref{ABrange2}) by the intermediate induction assumption, we have $0 \leq A-\lceil \log_2 (k+1) \rceil 2^{-\alpha}$ and $B+\lceil \log_2 (k+1) \rceil 2^{-\alpha} \leq 1$. Thus, we have 
\begin{align*}
    M_{\alpha,k} (x_1,\cdots,x_k), M_{\alpha,k+1} (&x_{k+1},\cdots,x_{2k+1}) \\
    &\in [A-\lceil \log_2 (k+1) \rceil 2^{-\alpha},B+\lceil \log_2 (k+1) \rceil 2^{-\alpha}] \\
    &\subseteq [0,1].
\end{align*}
Then, from Lemma (a) for $n=2$, we have
\begin{align*}
M_{\alpha,2k+1}(x_1, \cdots, x_{2k+1}) &= m_\alpha (M_{\alpha,k} (x_1,\cdots,x_k), M_{\alpha,k+1} (x_{k+1},\cdots,x_{2k+1})) \\
&\in [A-\lceil \log_2 (k+1) \rceil 2^{-\alpha} - 2^{-\alpha},B+\lceil \log_2 (k+1) \rceil 2^{-\alpha} + 2^{-\alpha}] \\
&= [A-\lceil \log_2 (2k+2) \rceil 2^{-\alpha},B+\lceil \log_2 (2k+2) \rceil 2^{-\alpha}]\\
&= [A-\lceil \log_2 (2k+1) \rceil 2^{-\alpha},B+\lceil \log_2 (2k+1) \rceil 2^{-\alpha}]
\end{align*}
Thus, Lemma (a) holds for $n=m=2k+1$. Then, Lemma (a) holds for all $n\geq 1$ by mathematical induction. \qed\\

(Lemma (b)) For $a,b,c,d \in \mathbb{R}$, we have
\begin{equation}\label{lemma_b}
|\max(a,b)-\max(c,d)| \leq \max(|a-c|,|b-d|)
\end{equation}

(Proof) Let $\max(a,b) = a$ without loss of generality. We denote the left-hand side and right-hand side of the inequality in (\ref{lemma_b}) by $\textup{LHS}$ and $\textup{RHS}$, respectively.
\begin{enumerate}
\item $\max(c,d) = c$\\
 We have $\textup{LHS} = |a-c|$. Thus, the lemma holds in this case.
\item $\max(c,d) = d$\\
We have $\textup{LHS} = |a-d|$. 
\begin{enumerate}
    \item $a \geq d$\\
    We have $a \geq d \geq c$. Thus, we have $\textup{LHS} = |a-d| \leq |a-c| \leq \textup{RHS}$. \item $a < d$\\
    We have $b \leq a < d$. Thus, we have $\textup{LHS} = |a-d| \leq |b-d| \leq \textup{RHS}$.
    
\end{enumerate}
\end{enumerate}
Thus, the lemma (b) is proved. \qed\\

Now, we prove the theorem using Lemma (a) and Lemma (b). We use mathematical induction to show that the inequality in (\ref{maxpool_ineq}) holds for all $n \in \mathbb{N}$. First, for $n=1$, we have
\[
|M_{\alpha,1}(x_1) - \max(x_1)| = |x_1 - x_1|=0=2^{-\alpha}\lceil \log_2 1 \rceil.
\]
Therefore, the inequality in (\ref{maxpool_ineq}) holds for $n=1$. 
We assume that the inequality holds for all $n$, $1 \leq n \leq m-1$.
Then, it is enough to show that the inequality in (\ref{maxpool_ineq}) also holds for $n = m$.
Suppose that $x_1,\cdots,x_{m}\in[(\lceil \log_2 m \rceil-1)2^{-\alpha},1-(\lceil \log_2 m \rceil-1)2^{-\alpha}] $.\\

(i) $m=2k$\\
We have to show that
\begin{equation}\notag
\begin{aligned}[b]
|M_{\alpha,2k}(x_1 , \cdots,x_{2k}) -& \max(x_1 , \cdots,x_{2k})| \leq 2^{-\alpha} \lceil \log_2 2k \rceil \\
&\mathrm{for}\:x_1,\cdots,x_{2k} \in [(\lceil \log_2 2k \rceil-1)2^{-\alpha},1-(\lceil \log_2 2k \rceil-1)2^{-\alpha}].
\end{aligned}
\end{equation}
We have 
$M_{\alpha,2k}(x_1,\cdots,x_{2k})=m_\alpha(M_{\alpha,k}(x_1,\cdots,x_k),M_{\alpha,k}(x_{k+1},\cdots,x_{2k})).$

Let $P=\max(x_1 , \cdots,x_{k})$, $Q=\max(x_{k+1} , \cdots,x_{2k})$, $\Tilde{P}=M_{\alpha,k}(x_1 , \cdots,x_{k})$, and $\Tilde{Q}=M_{\alpha,k}(x_{k+1} , \cdots,x_{2k})$.
Since 
\begin{align*}
x_1,\cdots,x_{2k} &\in [(\lceil \log_2 2k \rceil-1)2^{-\alpha},1-(\lceil \log_2 2k \rceil-1)2^{-\alpha}] \\
&\subseteq [(\lceil \log_2 k \rceil-1)2^{-\alpha},1-(\lceil \log_2 k \rceil-1)2^{-\alpha}],
\end{align*}
we can apply the intermediate induction assumption for $n=k$ as
\begin{align*}
|\Tilde{P}-P|=|M_{\alpha,k}(x_1 , \cdots,x_{k}) - &\max(x_1 , \cdots,x_{k})| \leq 2^{-\alpha} \lceil \log_2 k \rceil, \\
|\Tilde{Q}-Q|=|M_{\alpha,k}(x_{k+1} , \cdots,x_{2k}) - &\max(x_{k+1} , \cdots,x_{2k})| \leq 2^{-\alpha} \lceil \log_2 k \rceil.
\end{align*}
The left-hand side of the inequality in (\ref{maxpool_ineq}) for $n=m$ becomes $|m_\alpha (\Tilde{P}, \Tilde{Q})-\max(P,Q)|$. From Lemma (a) for $n=k$, we have $\Tilde{P}, \Tilde{Q} \in [0,1]$. Then, we have
\begin{align*}
|m_\alpha (\Tilde{P}, \Tilde{Q})-&\max(P,Q)| \leq |m_\alpha (\Tilde{P}, \Tilde{Q})-\max(\Tilde{P},\Tilde{Q})|+|\max(\Tilde{P},\Tilde{Q})-\max(P,Q)| \\
&\leq |m_\alpha (\Tilde{P}, \Tilde{Q})-\max(\Tilde{P},\Tilde{Q})| + \max(|\Tilde{P}-P|,|\Tilde{Q}-Q|) ~~~~\mbox{(from Lemma (b))}\\
&\leq 2^{-\alpha}+ \max(|\Tilde{P}-P|,|\Tilde{Q}-Q|) ~~~~\mbox{(from (\ref{m_alpha_condition}))} \\
&\leq 2^{-\alpha}+ 2^{-\alpha} \lceil \log_2 k \rceil \\ 
&=2^{-\alpha} \lceil \log_2 2k \rceil 
=2^{-\alpha} \lceil \log_2 m \rceil. 
\end{align*}

(ii) $m=2k+1$\\
We have to show that
\begin{align*}
|M_{\alpha,2k+1}(x_1 , \cdots,&x_{2k+1}) - \max(x_1 , \cdots,x_{2k+1})| \leq 2^{-\alpha} \lceil \log_2 (2k+1) \rceil \\
&\mathrm{for}\:x_1,\cdots,x_{2k+1} \in [(\lceil \log_2 (2k+1) \rceil-1)2^{-\alpha},1-(\lceil \log_2 (2k+1) \rceil-1)2^{-\alpha}].
\end{align*}

We have 
$M_{\alpha,2k+1}(x_1,\cdots,x_{2k+1})=m_\alpha(M_{\alpha,k}(x_1,\cdots,x_k),M_{\alpha,k+1}(x_{k+1},\cdots,x_{2k+1})).$

Let $P=\max(x_1 , \cdots,x_{k})$, $Q=\max(x_{k+1} , \cdots,x_{2k+1})$, $\Tilde{P}=M_{\alpha,k}(x_1 , \cdots,x_{k})$, and $\Tilde{Q}=M_{\alpha,k+1}(x_{k+1} , \cdots,x_{2k+1})$.
Since 
\begin{align*}
x_1,\cdots,x_k &\in [(\lceil \log_2 (2k+1) \rceil-1)2^{-\alpha},1-(\lceil \log_2 (2k+1) \rceil-1)2^{-\alpha}] \\
&\subseteq [(\lceil \log_2 k \rceil-1)2^{-\alpha},1-(\lceil \log_2 k \rceil-1)2^{-\alpha}] ~~\mbox{and}
\end{align*}
\begin{align*}
x_{k+1},\cdots,x_{2k+1} &\in [(\lceil \log_2 (2k+1) \rceil-1)2^{-\alpha},1-(\lceil \log_2 (2k+1) \rceil-1)2^{-\alpha}] \\
&\subseteq [(\lceil \log_2 (k+1) \rceil-1)2^{-\alpha},1-(\lceil \log_2 (k+1) \rceil-1)2^{-\alpha}],
\end{align*}
we can apply the induction assumption for $n=k$ and $n=k+1$ as
\begin{align*}
|\Tilde{P}-P|=|M_{\alpha,k}(x_1 , \cdots,x_{k}) - &\max(x_1 , \cdots,x_{k})| \leq 2^{-\alpha} \lceil \log_2 k \rceil, \\
|\Tilde{Q}-Q|=|M_{\alpha,k+1}(x_{k+1} , \cdots,x_{2k+1}) - &\max(x_{k+1} , \cdots,x_{2k+1})| \leq 2^{-\alpha} \lceil \log_2 (k+1) \rceil.
\end{align*}

The left-hand side of the inequality in (\ref{maxpool_ineq}) for $n=m$ becomes $|m_\alpha (\Tilde{P}, \Tilde{Q})-\max(P,Q)|$. From Lemma (a) for $n=k$ and $n=k+1$, we have $\Tilde{P}, \Tilde{Q} \in [0,1]$. Then, we have 
\begin{align*}
|m_\alpha (\Tilde{P}, \Tilde{Q})-&\max(P,Q)| \leq |m_\alpha (\Tilde{P}, \Tilde{Q})-\max(\Tilde{P},\Tilde{Q})|+|\max(\Tilde{P},\Tilde{Q})-\max(P,Q)| \\
&\leq |m_\alpha (\Tilde{P}, \Tilde{Q})-\max(\Tilde{P},\Tilde{Q})| + \max(|\Tilde{P}-P|,|\Tilde{Q}-Q|) ~~~~\mbox{(from Lemma (b))}\\
&\leq 2^{-\alpha}+ \max(|\Tilde{P}-P|,|\Tilde{Q}-Q|) ~~~~\mbox{(from (\ref{m_alpha_condition}))} \\
&\leq 2^{-\alpha}+ 2^{-\alpha} \lceil \log_2 (k+1) \rceil \\
&= 2^{-\alpha} \lceil \log_2 (2k+2) \rceil=2^{-\alpha} \lceil \log_2 (2k+1) \rceil
=2^{-\alpha} \lceil \log_2 m \rceil
\end{align*}
since $\lceil \log_2 (2k+2) \rceil=\lceil \log_2 (2k+1) \rceil$ for every integer $k\geq 1$.
Thus, the inequality in (\ref{maxpool_ineq}) holds for $n = m$, and the theorem is proved by mathematical induction.
\end{proof}

\paragraph{Proof of Theorem \ref{errorprop}.}
\begin{proof}
Equivalently, we have to show that 
\begin{equation}\label{claim_ineq}
    \|A_n^\alpha \circ \cdots \circ A_0^\alpha (\bb{x}) -  A_n \circ \cdots \circ  A_0 (\bb{x}) \|_{\infty} \leq (E_{A_n} ^{\alpha} \circ E_{A_{n-1}} ^{\alpha} \circ \cdots \circ E_{A_0} ^{\alpha})(0).
\end{equation}

We will prove it by mathematical induction. First, we will prove for $n=0$.
\begin{align*}
\|A_0^\alpha (\bb{x}) - A_0(\bb{x})\|_{\infty} \leq & \sup_{\|\bb{x}\|_\infty \leq B} \|A_0^\alpha (\bb{x}) - A_0 (\bb{x})\|_\infty \\
= &\sup_{\|\bb{x}+\bb{e}\|_\infty \leq B, \|\bb{e}\|_\infty \leq 0} \|A_0^\alpha (\bb{x}+\bb{e}) - A_0 (\bb{x})\|_\infty = E_{A_0}^\alpha (0).
\end{align*}
Thus, the inequality in (\ref{claim_ineq}) holds for $n=0$.

Next, we assume that the inequality in (\ref{claim_ineq}) holds for $n=k$, that is,
\[
\|A_k^\alpha \circ \cdots \circ A_0^\alpha (\bb{x}) - A_k \circ \cdots \circ  A_0 (\bb{x})\|_{\infty} \leq (E_{A_k} ^{\alpha} \circ E_{A_{k-1}} ^{\alpha} \circ \cdots \circ E_{A_0} ^{\alpha})(0)
\]
for some $k \geq 0$.
It is enough to show that the inequality in (\ref{claim_ineq}) also holds for $n=k+1$, that is, 
\[
\|A_{k+1}^\alpha \circ \cdots \circ A_0^\alpha (\bb{x}) - A_{k+1} \circ \cdots \circ  A_0 (\bb{x})\|_{\infty} \leq (E_{A_{k+1}} ^{\alpha} \circ E_{A_{k}} ^{\alpha} \circ \cdots \circ E_{A_0} ^{\alpha})(0).
\]

We have 
\begin{align*}
&\|A_{k+1}^\alpha \circ \cdots \circ A_0^\alpha (\bb{x}) - A_{k+1} \circ \cdots \circ  A_0 (\bb{x})\|_{\infty} \\
=& \|A_{k+1}^\alpha ( A_k \circ \cdots \circ A_0 (\bb{x}) + (A_k^\alpha \circ \cdots \circ A_1^\alpha (\bb{x}) - A_k \circ \cdots \circ A_0 (\bb{x})) ) - A_{k+1} \circ \cdots \circ  A_0 (\bb{x})\|_{\infty}.
\end{align*}
Let $\bb{x}' = A_k \circ \cdots \circ A_0 (\bb{x})$ and $\bb{e}' = A_k^\alpha \circ \cdots \circ A_0^\alpha (\bb{x}) - A_k \circ \cdots \circ A_0 (\bb{x})$.
Because $\|\bb{x}'+\bb{e}'\|_\infty = \|A_k^\alpha \circ \cdots \circ A_0^\alpha (\bb{x}) \|_\infty \leq B$,
we have 
\begin{align*}
\|A_{k+1}^\alpha \circ \cdots \circ A_0^\alpha (\bb{x}) - &A_{k+1} \circ \cdots \circ  A_0 (\bb{x})\|_{\infty} = \|A_{k+1}^\alpha (\bb{x}'+\bb{e}') - A_{k+1} (\bb{x}') \|_{\infty} \\
&\leq \sup_{\|\bb{x}+\bb{e}\|_\infty \leq B, \|\bb{e}\|_\infty \leq \|\bb{e}'\|_\infty} \| A_{k+1}^\alpha(\bb{x}+\bb{e}) - A_{k+1}(\bb{x}) \|_\infty \\
&= E_{A_{k+1}}^\alpha (\|\bb{e}'\|_\infty) \\
&= E_{A_{k+1}}^\alpha (\|A_k^\alpha \circ \cdots \circ A_1^\alpha (\bb{x}) - A_k \circ \cdots \circ A_1 (\bb{x})\|_{\infty})\\
&\leq E_{A_{k+1}}^\alpha (E_{A_k} ^{\alpha} \circ E_{A_{k-1}} ^{\alpha} \circ \cdots \circ E_{A_0} ^{\alpha})(0) ~~~~( \because E_{A_{k+1}}^\alpha(\cdot): ~\text{increasing})\\
&= E_{A_{k+1}} ^{\alpha} \circ E_{A_{k}} ^{\alpha} \circ \cdots \circ E_{A_0} ^{\alpha}(0).
\end{align*}
Thus, the inequality in (\ref{claim_ineq}) holds for $n=k+1$, and the theorem is proved by mathematical induction.

\end{proof}

\paragraph{Proof of Lemma \ref{basiclayers}.}
\begin{proof}
(a) For $A(\bb{x})=\bb{Ax}+\bb{b}$, $A^\alpha = A$. Then $A^\alpha(\bb{x}+\bb{e})-A(\bb{x})=\bb{Ae}$. Therefore, 
\begin{align*}
E_A^\alpha(e) &= \sup_{\|\bb{e}\|_\infty \leq e} \|\bb{A}\bb{e}\|_\infty 
 \\
 &\leq \sup_{\|\bb{e}\|_\infty \leq e} \|\bb{A}\|_\infty \|\bb{e}\|_\infty = \|\bb{A}\|_\infty e.
 \end{align*}

(b) For $A(\bb{x})=\ReLU(\bb{x})$, $A^\alpha(\bb{x})=\Tilde{r}_{\alpha,B} (\bb{x})$.
Therefore, 
\begin{align*}
E_A^\alpha(e) &= \sup_{\|\bb{x}+\bb{e}\|_\infty \leq B,\|\bb{e}\|_\infty \leq e} \|\Tilde{r}_{\alpha,B} (\bb{x} + \bb{e}) - \ReLU(\bb{x})\|_\infty \\
&\leq \sup_{\|\bb{x}+\bb{e}\|_\infty \leq B,\|\bb{e}\|_\infty \leq e} (\|\Tilde{r}_{\alpha,B} (\bb{x} + \bb{e}) - \ReLU(\bb{x}+\bb{e})\|_\infty + \| \ReLU(\bb{x}+\bb{e}) - \ReLU(\bb{x})\|_\infty ) \\
& \leq \sup_{\|\bb{x}+\bb{e}\|_\infty \leq B,\|\bb{e}\|_\infty \leq e} (B \cdot 2^{-\alpha}+\|\bb{e}\|_\infty) \\
&= B \cdot 2^{-\alpha}+e.
\end{align*}

(c) If $A$ is a max-pooling block with kernel size $k_0$, $\sup_{\|\bb{x}+\bb{e}\|_\infty \leq B,\|\bb{e}\|_\infty \leq e}\| A^\alpha (\bb{x}+\bb{e}) - A(\bb{x}) \|_\infty$ becomes \[
\sup_{\|\bb{x}+\bb{e}\|_\infty \leq B,\|\bb{e}\|_\infty \leq e} |\Tilde{M}_{\alpha,k_0^2,B}(x_1+e_1,\cdots,x_{k_0^2}+e_{k_0^2})-\max(x_1,\cdots,x_{k_0^2})|,
\]
where $\bb{x}=(x_1,\cdots,x_{k_0^2})$ and $\bb{e}=(e_1,\cdots,e_{k_0^2})$. 
Similar technique of proof for (b) and from Theorem \ref{maxpooling_thm}, we have $E_A^\alpha(e) \leq B'\lceil \log_2 k_0^2 \rceil 2^{-\alpha}  + e$, where $B' = B/(0.5-(\lceil \log_2 k_0^2 \rceil-1)2^{-\alpha})$.
Since $k_0 \leq 10$ and $\alpha \geq 4$, $(\lceil \log_2 k_0^2 \rceil-1)2^{-\alpha} <0.4$, therefore $B' <10B$, which leads the conclusion.

(d) Let us denote softmax block $A$ as an $\mathbb{R}^N \rightarrow \mathbb{R}^N$ function with $A(\bb{x})=(\frac{\exp(x_i)}{\sum_j \exp(x_j)})_{1\leq i \leq N}$ for $\bb{x}=(x_1,\cdots,x_N)$. Then mean value theorem for multivariate variable functions gives
\[
\|A(\bb{x}')-A(\bb{x})\|_\infty \leq \sup_{\bb{z}\in[\bb{x}',\bb{x}]}\|\bb{J}(\bb{z})\|_\infty\cdot \|\bb{x}'-\bb{x}\|_\infty,
\]
where $\bb{J}(\bb{z})$ denotes the Jacobian matrix of softmax block $A$. 
For a given vector $\bb{z}=(z_1,\cdots,z_N)\in \mathbb{R}^N$, the infinity norm of $\bb{J}(\bb{z})$ is given as $\|\bb{J}(\bb{z})\|_\infty = \max_i \sum_k |\frac{\partial}{\partial z_k } \frac{\exp(z_i)}{\sum_j \exp(z_j)}|$.
Note that
\[
\left | \frac{\partial}{\partial z_k } \frac{\exp(z_i)}{\sum_j \exp(z_j)} \right | = 
\begin{cases}
\cfrac{\exp(z_i)\exp(z_k)}{(\sum_j \exp(z_j))^2}, & k \neq i\\
\cfrac{\exp(z_i)\sum_{j\neq i}\exp(z_j)}{(\sum_j \exp(z_j))^2}, & k = i
\end{cases}
\]
and thus
\begin{align*}
\sum_k \left |\frac{\partial}{\partial z_k } \frac{\exp(z_i)}{\sum_j \exp(z_j)} \right | = \frac{\exp(z_i)\sum_{k\neq i}\exp(z_k)}{(\sum_j \exp(z_j))^2}+ \frac{\exp(z_i)\sum_{j\neq i}\exp(z_j)}{(\sum_j \exp(z_j))^2} =2p_i(1-p_i),
\end{align*}
where $p_i=\frac{\exp(z_i)}{\sum_j \exp(z_j)}$. Therefore, $\|\bb{J}(\bb{z})\|_\infty = \max_i 2p_i(1-p_i) \leq 1/2 $ for any vector $\bb{z}\in\mathbb{R}^N$, and this gives
\[
\|A^\alpha(\bb{x}+\bb{e})-A(\bb{x})\|_\infty=\|A(\bb{x}+\bb{e})-A(\bb{x})\|_\infty  \leq \|\bb{e}\|_\infty/2,
\]
which leads the conclusion.
\end{proof}

\paragraph{Proof of Theorem \ref{main}.}
\begin{proof}
From Theorem \ref{errorprop}, it is enough to show that 
\begin{equation}\label{thm4_claim_ineq}
(E_{A_n} ^{\alpha} \circ \cdots \circ E_{A_0} ^{\alpha})(0) \leq C 2^{-\alpha}
\end{equation}
for some constant $C$ and $n \geq 0$. To show that this inequality holds for all $n \geq 0$, we use mathematical induction.
First, for $n=0$, assume that $\mathcal{F}$ has one block $A$.
The upper bounds of $E_A^{\alpha}(0)$ for the four basic blocks suggested in Lemma \ref{basiclayers} have forms of $C_0 \cdot 2^{-\alpha}$, where $C_0$ can be zero. Thus, the inequality in (\ref{thm4_claim_ineq}) holds for some constant $C$ and $n=0$.

Then, we assume that the inequality in (\ref{thm4_claim_ineq}) holds for $n=k$ and some constant $C_k$, that is, $(E_{A_k} ^{\alpha} \circ \cdots \circ E_{A_0} ^{\alpha})(0) \leq C_k 2^{-\alpha}$ for some $C_k$. 
Then, it is enough to show that the inequality in (\ref{thm4_claim_ineq}) holds for $n=k+1$ and some constant $C_{k+1}$. $E_{A_{k+1}}^\alpha(C_k 2^\alpha)$ is not greater than 
\begin{enumerate}[label=(\roman*)]
\item $\| \bb{A} \|_\infty C_k 2^{-\alpha}$ when $A_{k+1}$ is a linear block with $A(\bb{x})=\bb{Ax}+\bb{b}$, 
\item $(B + C_k) 2^{-\alpha}$ when $A_{k+1}$ is a ReLU block, 
\item $(10B\lceil \log_2 k_0^2 \rceil + C_k) 2^{-\alpha}$ when $A_{k+1}$ is a max-pooling block with kernel size $k_0$,
\item $\frac{1}{2}C_k 2^{-\alpha}$ when $A_{k+1}$ is a softmax block
\end{enumerate}
by Lemma \ref{basiclayers}. For each case, we can determine a constant $C_{k+1}$ that satisfies $(E_{A_{k+1}} ^{\alpha} \circ \cdots \circ E_{A_0}^\alpha) (0) \leq E_{A_{k+1}}^\alpha(C_k 2^\alpha) \leq C_{k+1} 2^{-\alpha}$. Thus, the theorem is proved.
\end{proof}

\section{Generalization of Theorem 4 for ResNet model}
ResNet model cannot be decomposed of basic blocks, since it contains a ``residual block''. The authors of \cite{he2016deep} suggest an operation $R$ that satisfies
\[
R(\bb{x}) = \mathcal{G}(\bb{x}) + \bb{P} \bb{x},
\]
where $\mathcal{G}(\bb{x})$ denotes the residual mapping which will be learned, and $\bb{P}$ is a linear projection matrix to match the dimension of $\mathcal{G}(\bb{x})$. 
In this section, we call such operation $R$ as a \textit{residual block}. 
In the residual block designed in ResNet, all residual mapping $\mathcal{G}(\bb{x})$ is a composition of basic blocks \cite{he2016deep}. 
Therefore, by Theorem \ref{main}, we can determine a constant $C_\mathcal{G}$ that satisfies $\|\mathcal{G}^{\alpha} (\bb{x}) - \mathcal{G}(\bb{x})\|_\infty \leq C_\mathcal{G} 2^{-\alpha}$. In case of $\bb{P}$, there are two methods of constructing  projection in ResNet:
the first one pads extra zeros, and the second one uses $1 \times 1$ convolution \cite{he2016deep}.
We note that both methods can be considered as linear blocks, and thus the approximate block $R^\alpha(\bb{x})$ can be represented by $\mathcal{G}^\alpha(\bb{x}) + \bb{Px}$.
Considering this situation, we generalize the original Theorem \ref{main} so that it is also valid for ResNet model.

\begin{theorem}[generalized version of Theorem \ref{main}]\label{generalized}
If $\alpha \geq 4$ and a deep learning model $\mathcal{F}$ contains only basic blocks and residual blocks, then there exists a constant $C$ such that $\|\mathcal{F}^{\alpha} (\bb{x}) - \mathcal{F}(\bb{x})\|_\infty \leq C2^{-\alpha}$ for every input $\bb{x}$, where the constant $C$ can be determined only by model parameters. 
\end{theorem}
\begin{proof}
We prove this statement by obtaining error propagation function $E_R^\alpha (e)$ of residual block $R(\bb{x})=\mathcal{G}(\bb{x}) + \bb{P} \bb{x}$. From the definition of error propagation function, we have
\begin{align*}
E_R^\alpha (e) &= \sup_{\|\bb{x}+\bb{e}\|_\infty \leq B,\|\bb{e}\|_\infty \leq e}\| R^\alpha(\bb{x}+\bb{e}) -R(\bb{x})\|_\infty \\
&\leq \sup_{\|\bb{x}+\bb{e}\|_\infty \leq B,\|\bb{e}\|_\infty \leq e}(\| R^\alpha(\bb{x}+\bb{e}) -R(\bb{x}+\bb{e})\|_\infty  + \| R(\bb{x}+\bb{e}) -R(\bb{x})\|_\infty ) \\
&\leq \sup_{\|\bb{x}+\bb{e}\|_\infty \leq B,\|\bb{e}\|_\infty \leq e}(\| \mathcal{G}^\alpha(\bb{x}+\bb{e}) -\mathcal{G}(\bb{x}+\bb{e})\|_\infty + \| \mathcal{G}(\bb{x}+\bb{e}) -\mathcal{G}(\bb{x})+\bb{Pe}\|_\infty ) \\
&\leq C_\mathcal{G} 2^{-\alpha} + \| \bb{P} \|_\infty e + \sup_{\|\bb{x}+\bb{e}\|_\infty \leq B,\|\bb{e}\|_\infty \leq e}  \| \mathcal{G}(\bb{x}+\bb{e}) -\mathcal{G}(\bb{x}) \|_\infty,
\end{align*}
where $C_\mathcal{G}$ is the constant that satisfies $\|\mathcal{G}^{\alpha} (\bb{x}) - \mathcal{G}(\bb{x})\|_\infty \leq C_\mathcal{G} 2^{-\alpha}$. To estimate $ \| \mathcal{G}(\bb{x}+\bb{e}) -\mathcal{G}(\bb{x}) \|_\infty$, we decompose residual mapping $\mathcal{G}$ into $G_k \circ \cdots \circ G_0$, where $G_i$'s are basic blocks. 
Then, we define
\[
\Delta_A(e) = \sup_{\|\bb{x}+\bb{e}\|_\infty \leq B,\|\bb{e}\|_\infty \leq e} \| A(\bb{x}+\bb{e})-A(\bb{x}) \|_\infty
\]
for a block $A$ and magnitude of the error, $e$. Then $\Delta_A(e)$ is a non-decreasing function of $e$ for every block $A$. 
Thus, similar argument in the proof of Theorem \ref{errorprop} shows that
\begin{equation}\label{delprop}
\|\mathcal{G}(\bb{x}+\bb{e}) - \mathcal{G}(\bb{x}) \|_\infty \leq (\Delta_{G_k} \circ \cdots \Delta_{G_0})(\| \bb{e} \|_\infty)
\end{equation}
for every error vector $\bb{e}$.
Also, similar argument in the proof of Lemma \ref{basiclayers} shows that
\begin{equation}\label{delbound}
\Delta_A(e) \leq
\begin{cases}
\| \bb{A} \|_\infty e, & A: \text{linear block, }A(\bb{x}) = \bb{Ax}+\bb{b} \\
e, & A:\text{ReLU} \\
e, & A:\text{max-pooling} \\
e/2 & A:\text{softmax} \\
\end{cases}
\end{equation}
for every basic block $A$ and $e\geq 0$. These inequalities for $\Delta_A (e)$ corresponds to the inequalities for $E_A^\alpha (e)$ in Lemma 1 when the precision parameter $\alpha$ goes to infinity. 
From two inequalities (\ref{delprop}) and (\ref{delbound}), there exists a constant $C_\mathcal{G}'$ that satisfies $\| \mathcal{G} (\bb{x}+\bb{e})-\mathcal{G}(\bb{x})\|_\infty \leq C_\mathcal{G}' \|\bb{e}\|_\infty$ since $G_i$'s are all basic blocks. 
Therefore, we have
\[
E_R^\alpha (e) \leq C_\mathcal{G} 2^{-\alpha} + (\| \bb{P} \|_\infty +  C_\mathcal{G}' )e 
\]
for every $e \geq 0$. 

Now, we prove the theorem using mathematical induction.
Let $\mathcal{F}$ be a deep learning model which can be decomposed into $A_n \circ \cdots \circ A_0$ where $A_k$'s are all basic blocks or residual blocks. 
First, for $n=0$, assume that $\mathcal{F}$ has one block $A$. The upper bounds of $E_A^\alpha(0)$ for the four basic blocks and residual blocks have forms of $C_0 2^{-\alpha}$, where $C_0$ can be zero.
Inductively, for $n=k$, $(E_{A_k} ^{\alpha} \circ \cdots \circ E_{A_0} ^{\alpha})(0) \leq C_k 2^{-\alpha}$ for some $C_k$. 
For $n=k+1$, we can determine a constant $C_{k+1}$ that satisfies $(E_{A_{k+1}} ^{\alpha} \circ \cdots \circ E_{A_0} ^{\alpha})(0) \leq C_{k+1} 2^{-\alpha}$ when $A_{k+1}$ is a basic block (see the proof of Theorem \ref{main}).
If $A_{k+1}$ is a residual block with $A_{k+1}(\bb{x}) = \mathcal{G}(\bb{x}) + \bb{Px}$ for some residual mapping $\mathcal{G}(\bb{x})$ and linear projection matrix $\bb{P}$, then
\[
E_{A_{k+1}}^\alpha(C_k 2^\alpha) \leq C_\mathcal{G} 2^{-\alpha} + (\| \bb{P} \|_\infty +  C_\mathcal{G}' )(C_k 2^\alpha) = C_{k+1}2^{-\alpha}
\]
where $C_{k+1}=C_\mathcal{G} + (\| \bb{P} \|_\infty +  C_\mathcal{G}' )C_k$ which is independent of $\alpha$. Therefore, 
\[
(E_{A_{k+1}} ^{\alpha} \circ \cdots \circ E_{A_0} ^{\alpha})(0) \leq E_{A_{k+1}}^\alpha(C_k 2^\alpha) \leq C_{k+1} 2^{-\alpha},
\]
which completes the proof.
\end{proof}

\end{document}